\documentclass[prd,aps,eqsecnum,superscriptaddress,nofootinbib,notitlepage,longbibliography,draft]{revtex4-2}

\usepackage[a4paper, left=2cm, right=2cm, top=2cm, bottom=2cm]{geometry}
\usepackage{setspace}
\usepackage[english]{babel}
\addto\extrasenglish{%
}
\usepackage[utf8]{inputenc}
\usepackage[T1]{fontenc}
\usepackage{lmodern}
\usepackage{amsmath,amsthm,mathdots,amssymb,bm,bbm,mathrsfs,mathtools}
\usepackage[final]{graphicx} 
\usepackage{subcaption}
\usepackage{tikz,pgfplots}\pgfplotsset{compat=1.16}

\usepackage{thm-restate}
\newtheorem{thm}{Theorem}[section]
\newtheorem{cor}{Corollary}[section]
\newtheorem{lem}{Lemma}[section]

\newtheorem{rmk}{Remark}[thm]
\newtheorem{defn}{Definition}[section]

\usepackage{comment}

\usepackage[colorlinks=true, urlcolor=violet, linkcolor=blue, citecolor=red, hyperindex=true, linktocpage=true, pagebackref=true, draft=false]{hyperref}
\usepackage{mleftright}\mleftright 

\usepackage{enumitem}
\mathtoolsset{showonlyrefs=true}
\usepackage{microtype,color,graphicx} 
\usepackage{tikz-cd}
\usepackage{xcolor}
\usepackage{multirow}
\usepackage{tablefootnote}

\renewcommand{\vec}{\bm}

\newcommand{\CC}{\mathcal{C}}

\newcommand{\CD}{\mathcal{D}}
\newcommand{\CE}{\mathcal{E}}
\newcommand{\BE}{\mathbb{E}}

\newcommand{\CI}{\mathcal{I}}

\newcommand{\CL}{\mathcal{L}}

\newcommand{\CO}{\mathcal{O}}
\newcommand{\CP}{\mathcal{P}}

\newcommand{\CR}{\mathcal{R}}
\newcommand{\BR}{\mathbb{R}}

\newcommand{\CT}{\mathcal{T}}

\newcommand{\vA}{\bm{A}}

\newcommand{\vB}{\bm{B}}

\newcommand{\vC}{\bm{C}}

\newcommand{\vF}{\bm{F}}

\newcommand{\vh}{\bm{h}}

\newcommand{\vH}{\bm{H}}
\newcommand{\vI}{\bm{I}}

\newcommand{\vO}{\bm{O}}
\newcommand{\vP}{\bm{P}}

\newcommand{\vS}{\bm{S}}

\newcommand{\vU}{\bm{U}}
\newcommand{\vV}{\bm{V}}

\newcommand{\vX}{\bm{X}}
\newcommand{\vY }{\bm{Y }}
\newcommand{\vZ}{\bm{Z}}
\newcommand{\vsigma}{\bm{ \sigma}}

\newcommand{\vrho}{\bm{ \rho}}
\renewcommand{\L}{\left}
\newcommand{\R}{\right}

\newcommand{\dagg}{\dagger}

\newcommand{\vertiii}[1]{{\left\vert\kern-0.25ex\left\vert\kern-0.25ex\left\vert #1 \right\vert\kern-0.25ex\right\vert\kern-0.25ex\right\vert}}
\newcommand{\norm}[1]{\Vert {#1} \Vert}

\newcommand{\normp}[2]{\norm{#1}_{#2}}
\newcommand{\lnormp}[2]{\lnorm{#1}_{#2}}

\newcommand{\labs}[1]{\left\vert {#1} \right\vert}
\newcommand{\lnorm}[1]{\left\Vert {#1} \right\Vert}
\newcommand{\e}{\mathrm{e}}

\newcommand{\ri}{\mathrm{i}}
\newcommand{\rd}{\mathrm{d}}
\newcommand*{\tr}{\mathrm{Tr}}
\newcommand*{\poly}{\mathrm{Poly}}

\newcommand*{\Supp}{\mathrm{Supp}}
\newcommand{\indicator}{\mathbbm{1}}

\newcommand{\nrm}[1]{\left\| #1 \right\|}

\newcommand{\undersetbrace}[2]{ \underset{#1}{\underbrace{#2}}}

\DeclarePairedDelimiterX{\braket}[1]{\langle}{\rangle}{#1}
\DeclarePairedDelimiterX\ketbra[2]{| }{|}{#1 \delimsize\rangle\!\delimsize\langle #2}	
\DeclarePairedDelimiterX\dotp[2]{\langle}{\rangle}{#1, #2}

\DeclareMathAlphabet{\dutchcal}{U}{dutchcal}{m}{n}
\SetMathAlphabet{\dutchcal}{bold}{U}{dutchcal}{b}{n}
\DeclareMathAlphabet{\dutchbcal} {U}{dutchcal}{b}{n}

\makeatletter
\DeclareRobustCommand*{\pmzerodot}{%
	\nfss@text{%
		\sbox0{$\vcenter{}$}
		\sbox2{0}%
		\sbox4{0\/}%
		\ooalign{%
			0\cr
			\hidewidth
			\kern\dimexpr\wd4-\wd2\relax 
			\raise\dimexpr(\ht2-\dp2)/2-\ht0\relax\hbox{%
				\if b\expandafter\@car\f@series\@nil\relax
				\mathversion{bold}%
				\fi
				$\cdot\m@th$%
			}%
			\hidewidth
			\cr
			\vphantom{0}
		}%
	}%
}

\usepackage{ifdraft}
\ifdraft{
	\newcommand{\authnote}[3]{{\color{#3} {\bf  #1:} #2}}	
}{
	\newcommand{\authnote}[3]{}
}

\usetikzlibrary{quantikz}
\begin{document}

\title{Quantum Gibbs states are locally Markovian}

 \author{Chi-Fang Chen}
 \email{achifchen@gmail.com}
 \affiliation{University of California, Berkeley, CA, USA }
\affiliation{Massachusetts Institute of Technology, Cambridge, USA}

\author{Cambyse Rouz\'{e}}
\email{cambyse.rouze@inria.fr}
 \affiliation{Inria, Télécom Paris - LTCI,\\ Institut Polytechnique de Paris, 91120 Palaiseau, France}

\begin{abstract}
The Markov property entails the conditional independence structure inherent in Gibbs distributions for general classical Hamiltonians, a feature that plays a crucial role in inference, mixing time analysis, and algorithm design. However, much less is known about quantum Gibbs states. In this work, we show that for any Hamiltonian with a bounded interaction degree, the quantum Gibbs state is locally Markov at arbitrary temperature, meaning there exists a quasi-local recovery map for every local region. Notably, this recovery map is obtained by applying a detailed-balanced Lindbladian with jumps acting on the region. Consequently, we prove that (i) the conditional mutual information (CMI) for a shielded small region decays exponentially with the shielding distance, and (ii) under the assumption of uniform clustering of correlations, Gibbs states of general non-commuting Hamiltonians on $D$-dimensional lattices can be prepared by a quantum circuit of depth $\e^{\mathcal{O}(\log^D(n/\epsilon))}$, which can be further reduced assuming certain local gap condition. Our proofs introduce a regularization scheme for imaginary-time-evolved operators at arbitrarily low temperatures and reveal a connection between the Dirichlet form, a dynamic quantity, and the commutator in the KMS inner product, a static quantity. We believe these tools pave the way for tackling further challenges in quantum thermodynamics and mixing times, particularly in low-temperature regimes.
\end{abstract}
\maketitle
\tableofcontents
\section{Introduction}
We say three correlated random variables $ABC$ are Markov if $A$ and $C$ are independent conditioned on $B$. Such \textit{Markov property} is one of the most succinct yet versatile structural properties, allowing us to decompose correlated, high-dimensional distributions into conditionally independent sets of variables. In particular, in Markov random fields or graphical models, the conditional dependence relations can be effectively captured as edges on graphs, providing a starting point for computing marginals, maximum a posteriori assignments, and the partition function \cite{Stefankovic09,Bauerschmidt2019}. Such \textit{static} conditional dependence structure also ties deeply with local detailed-balanced Metropolis or Glauber \textit{dynamics}, which, if mixed rapidly, enables efficient sampling from the distribution~\cite{Eldan2021,Anari2022,Chen2021,anari2021entropic}. Remarkably, despite such powerful structural constraints, Markov random fields remain expressive and rich, encompassing many combinatorial optimization and counting problems~\cite{wainwright2008graphical,koller2009probabilistic,Mzard2009,Jerrum1995,Stefankovic09,https://doi.org/10.1184/r1/6477104.v1,blajerrum}; in fact, due to the Hammersley-Clifford theorem \cite{hammersleymarkov}, Markov random fields are \textit{equivalent} to Gibbs distributions of many-body classical Hamiltonians, which underpins the vast majority of classical statistical physics problems (\autoref{fig:Ising_Markov}). 

In quantum many-body physics and quantum computation, quantum Gibbs states play a similarly essential role, encapsulating
many-body quantum systems in thermal equilibrium. Given a set $\Lambda$ of $n=|\Lambda|$ qubits and a few-body Hamiltonian
\begin{align*}
\vH=\sum_{\gamma\in\Gamma}\vH_\gamma\quad \text{where}\quad \|\vH_\gamma\|\le 1,
\end{align*}
 the quantum Gibbs state corresponding to $\vH$ at inverse temperature $\beta>0$ is defined as
\begin{align*}
\vrho_\beta:=\e^{-\beta\vH}/\tr(\e^{-\beta\vH})\,.
\end{align*}
Much progress in quantum information considers general structural properties of quantum Gibbs states or ground states. Notable examples include the area law of entanglement and tensor networks~\cite{DMRG,Hastings2007AnAL,Garcia07,Vestraete2008}, which provide an efficient parameterization of physically motivated quantum states and lead to celebrated physical and algorithmic consequences. However, the general structure and complexity of quantum Gibbs states remain largely speculative. In finely-tuned cases, quantum Gibbs states could encode classically challenging computational problems for commuting Hamiltonians~\cite{bergamaschi2024quantum, https://doi.org/10.48550/arxiv.2408.01516} or at low temperatures~\cite{chen2024local, rouze2024efficient}. At high enough temperatures or in 1D, Gibbs states often have efficient classical algorithms~\cite{kuwahara2018polynomial, fawzi2023certified, harrow2020classical, fawzi2022subpolynomial, Helmuth2023} and 
exhibit classical product state behavior \cite{bakshi2024high}.

Nevertheless, the most celebrated classical Markov properties are, strictly speaking, \textit{false} for quantum Gibbs states. Given a tripartition of the system $\Lambda=A\sqcup B\sqcup C$, the \textit{quantum conditional mutual information} (QCMI) between $A$ and $C$ conditioned on $B$ for a quantum state $\vrho$ is defined as
\begin{align*}
I(A:C|B)_{\vrho}:=S(\vrho_{AB})+S(\vrho_{BC})-S(\vrho_B)-S(\vrho_{ABC})\, ,
\end{align*}
with $S(\vsigma):=-\tr(\vsigma\log\vsigma)$ being the entropy of a state $\vsigma$. Such a tripartite quantity is an indicator for quantum Markovianity because a vanishing QCMI implies a recovery map $\CR_{AB}$ acting only on subsystem $AB$ that recovers a lost system $A$~\cite{hayden2004structure,brown2012quantum}
\begin{equation}\tag{Exact Markov}\label{introeq.Markov}
I(A:C|B)_{\vrho}=0\quad \iff\quad    \exists \CR_{AB},\quad \CR_{AB}[\vrho_{BC}] = \vrho.
\end{equation}
Unfortunately, for quantum Gibbs states with tripartition $ABC$ such that $B$ shields $A$ from $C$, the QCMI may not vanish exactly for general Hamiltonians except for the commuting case. For decades, it has been a major open problem to prove approximate versions of the decay of QCMI as a step toward a suitable quantum Hammersley-Clifford theorem (see~\autoref{sec:intro_approx_Markov}). 

Recently, there has been a wave of new quantum Gibbs sampling algorithms serving as the quantum analog of classical Metropolis~\cite{temme2011quantum,yung2012quantum,gilyen2024quantum} or Glauber dynamics~\cite{Shtanko2021AlgorithmsforGibbs, chen2021fast, rall2023thermal, wocjan2023szegedy,chen2023quantum,jiang2024quantum,ding2024single,ding2024efficient}. In particular,~\cite{chen2023efficient} gives a (quasi)-local and detailed balanced Lindbladian $\CL$ that fixes the quantum Gibbs state
\begin{align}
    \CL = \sum_{a} \CL_a\quad \text{where} \quad \CL_a[\vrho_{\beta}] = 0 \quad \text{for each jump}\quad a.
\end{align}
The associated mixing time quantifies the convergence rate of such dynamics, whose classical analog has deep ties with the static properties of the underlying distribution \cite{martinelli1999lectures}. This static/dynamic equivalence was further extended to commuting Hamiltonians in a series of works~\cite{capel2021modified,Bardet2021EntropyDF,Bardet2024,kochanowski2024rapid,capel2024quasi}. Recently, optimal quantum mixing times at high temperatures were proven for general Hamiltonians~\cite{rouze2024optimal, rouze2024efficient}, giving efficient quantum Gibbs sampling algorithms in that regime. Given the intimate interplay between the Markov property and classical sampling algorithms, one naturally wonders whether quantum Gibbs sampling could offer a dynamical angle to the static problem at hand. The following themes guide this paper:
\smallskip

\begin{center}
\textit{ When are Lindbladians good recovery maps? }
\medskip

\textit{How do mixing times interact with Markov properties?}
\end{center}
In a nutshell, our main results show that quantum Gibbs states are locally Markov: for each local region, we construct a quasi-local recovery map (\autoref{fig:recovery_animation}). Remarkably, such a property holds at any temperature for any Hamiltonian with a bounded interaction degree, independent of thermal phase transitions. In particular, our recovery map is precisely running a Linbladian Gibbs sampler with jumps supported on that region.

\begin{figure}[t]
\includegraphics[width=0.95\textwidth]{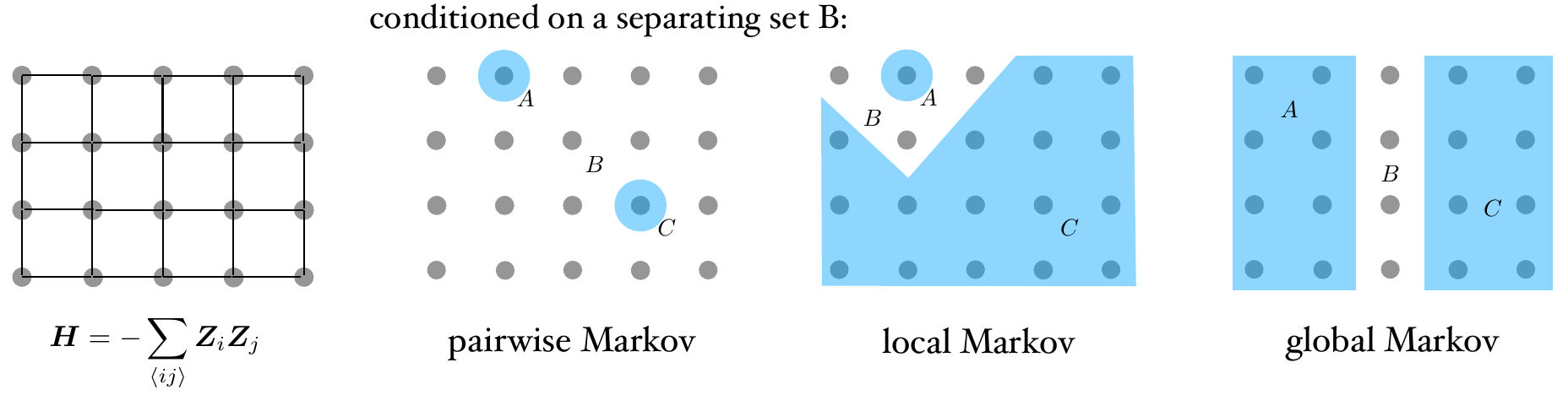}\caption{ 
Classical Gibbs distributions of local Hamiltonians are always Markovian at any non-zero temperatures. For the 2D Ising model with nearest neighbour interactions, the regions $A$ and $C$ are independent conditioned on the separating set $B$. Precise notions of pairwise Markov property (when $\labs{A},\labs{C}= 1$), local Markov property (when $\labs{A}=1$), and global Markov property (when $A$ is any subset) can be defined, and, are equivalent for positive distributions. In the quantum case, these three are not known to be equal, even allowing for approximations. 
}\label{fig:Ising_Markov}
\end{figure}

\subsection{Approximate Markov properties}
\label{sec:intro_approx_Markov}
\begin{table}
\centering
\begin{tabular}{|l|l|l|} 
\hline
 Markov property & Sufficient conditions  & QCMI $I(A:C|B)_{\vrho_{\beta}}$ for $ABC=\Lambda$ \\ \hline
Pairwise & $D$-dim lattice at any $\beta>0$~\cite{kuwahara2024clustering} & $\exp\bigg(c\labs{AC}-\operatorname{dist}(A,C)/\xi\bigg)$    \\ \hline
Local & Degree $d$ at any $\beta>0$ (\autoref{thm:main}) 
& $\labs{A}\labs{C}\exp\bigg(c \min(\labs{A},\labs{C}) -\operatorname{dist}(A,C)/\xi\bigg)$  \\ \hline
Global & Commuting or classical at any $\beta>0$ & $0$ if $\operatorname{dist}(A,C)\ge 1$ \\ 
\cline{2-3}
          & 1-dim at any $\beta>0$~\cite{Kato2019,kuwahara2024clustering}                       & $\exp\L(- \operatorname{dist}(A,C)/\xi\R)$\\ \hline
\end{tabular}
\caption{\label{table:summary} Known criteria for Quantum Markov properties to hold at varying generality. The distance $\operatorname{dist}(A,C)$ between sets $A$ and $C$ is defined as in~\autoref{sec:Ham} for general connectivity and reduces to Euclidean distance in $D$-dimensional lattices (up to multiplicative constants). The scalar $c$ and $\xi$ are constants that may vary from line to line, but only depend on inverse temperature $\beta$ and the connectivity of the Hamiltonian (i.e., the dimension $D$ or the interaction degree $d$). While all three types of Markov properties feature CMI decaying exponentially with the distance $\operatorname{dist}(A,C)$, the extra dependence on the sizes of $\labs{A},\labs{C}$ determines how large $A$ and $C$ can meaningfully be.}
\end{table}

In the recent years, various attempts at recovering an approximate version of the quantum Markov property \eqref{introeq.Markov} have been made:

\medskip
\noindent \textbf{Conjecture} \cite{kuwahara2020clustering}: given the Gibbs state $\vrho_\beta$ of a short-range interaction Hamiltonian on a $D$-dimension lattice at any inverse temperature $\beta$, and any tripartition of the system $\Lambda=A\sqcup B\sqcup C$, the quantum conditional mutual information (QCMI) evaluated at $\vrho_\beta$ satisfies
\begin{align}
    I(A:C|B)_{\vrho_\beta}\le \mathcal{D}(\operatorname{dist}(A,C))
\end{align}
for some superpolynomially decaying function $\mathcal{D}$ of the distance between $A$ and $C$, which also depends on $\beta$ and the geometry of $A,B,C$.

\medskip

Essentially, three main versions of the conjecture were considered in the literature \cite{kuwahara2024clustering}, listed below in order of increasing strength \cite{lauritzen1996graphical,koller2009probabilistic} (see~\autoref{table:summary}):

\begin{itemize}
\item \textit{Pairwise Markov property}: both subsystems A and C must be small: $|A|,|C| =\mathcal{O}(1)$.
\item \textit{Local Markov property}: $A$ or $C$ must be small, i.e., $\min(|A|, |C|) = \mathcal{O}(1)$.
\item \textit{Global Markov property}: both $A$ and $C$ can be macroscopic, i.e., $|A|, |C| = \mathcal{O}(|\Lambda|)$.
\end{itemize}
In \cite{Kato2019}, the authors proved the \textit{global Markov property} in the restricted case of $1D$ quantum spin chains. They derived a subexponentially decaying function $\mathcal{D}$ by leveraging quantum belief propagation equations, Lieb-Robinson bounds, and the exponential decay of correlations for $1D$ systems \cite{araki1969gibbs}. This bound was later strengthened into an exponential decay in \cite{kuwahara2024clustering}.
 In the general $D$-dimensional setting, the impressive work~\cite{kuwahara2024clustering} proved the \textit{pairwise Markov property} by constructing an effective Hamiltonian for the reduced states over subregions of the lattice.

In this paper, we prove that Gibbs states of Hamiltonians with short-range interaction at any temperature on $D$-dimensional lattices satisfy the \textit{local Markov property}, setting an open problem of~\cite{kuwahara2024clustering}. More precisely, 
for any tripartition of the system $\Lambda=A\sqcup B\sqcup C$ the quantum conditional mutual information (QCMI) evaluated at $\vrho_\beta$ satisfies
\begin{align}\label{CMIdecayintro}
    I(A:C|B)_{\vrho_\beta}\le r |A||C|\exp\bigg(c \min(\labs{A},\labs{C}) -\frac{\operatorname{dist}(A,C)}{\xi}\bigg)\,,
\end{align}
for some positive numbers $r,c,\xi$ depending on $\beta$ and $D$. In fact, our results were proved in the more general setting of interaction graphs with bounded degree, mirroring the setting of classical graphical models (see~\autoref{sec:main results}). However, finding a QCMI bound which depends only polynomially on $|A|$ and $|C|$, namely a \textit{global Markov property}, remains open.\footnote{In \cite{kuwahara2020clustering}, an even stronger notion of clustering of the QCMI at high enough temperature and for subregions $ABC$ of the lattice $\Lambda=ABCD$, was considered using high temperature cluster expansions. However, the proof requires expansions of operator-valued partial trace functionals, whose correctness remains unclear.} 

To prove the local Markov property, we explicitly construct a recovery map $\mathcal{R}_{AB}$ such that
\begin{align*}
\|\mathcal{R}_{AB}(\vrho_{BC})-\vrho\|_1\le\, r'\,\exp\left(c'|A|-\frac{\operatorname{dist}(A,C)}{\xi'}\right) 
\end{align*}
for some positive numbers $r',c',\xi'$ depending on $\beta$ and the dimension $D$. In fact, the recovery map arises from applying a time-averaged Lindbladian dynamics
\begin{align}
    \CR_{A,t}[\cdot]:= \frac{1}{t}\int_{0}^t \exp\L(s\,\sum_{a\in P^1_A}\CL_a\R)[\cdot] \,\rd s\label{eq:RAt_intro}
\end{align}
for generators $\CL_a$ with jumps $\vA^a \in P^1_A$ being all single site Pauli operators ($\vX,\vZ,\vZ$) acting on the region $A$. Due to Lieb-Robinson bounds, the Lindbladian depends mostly on the local Hamiltonian patch.

When this manuscript is near completion, we become aware of the concurrent and independent results of Kohtaro and Kuwahara which prove a stronger global Markov property, but only at high temperatures, also using certain detailed balanced dynamics. 

\begin{figure}[t]
\includegraphics[width=0.9\textwidth]{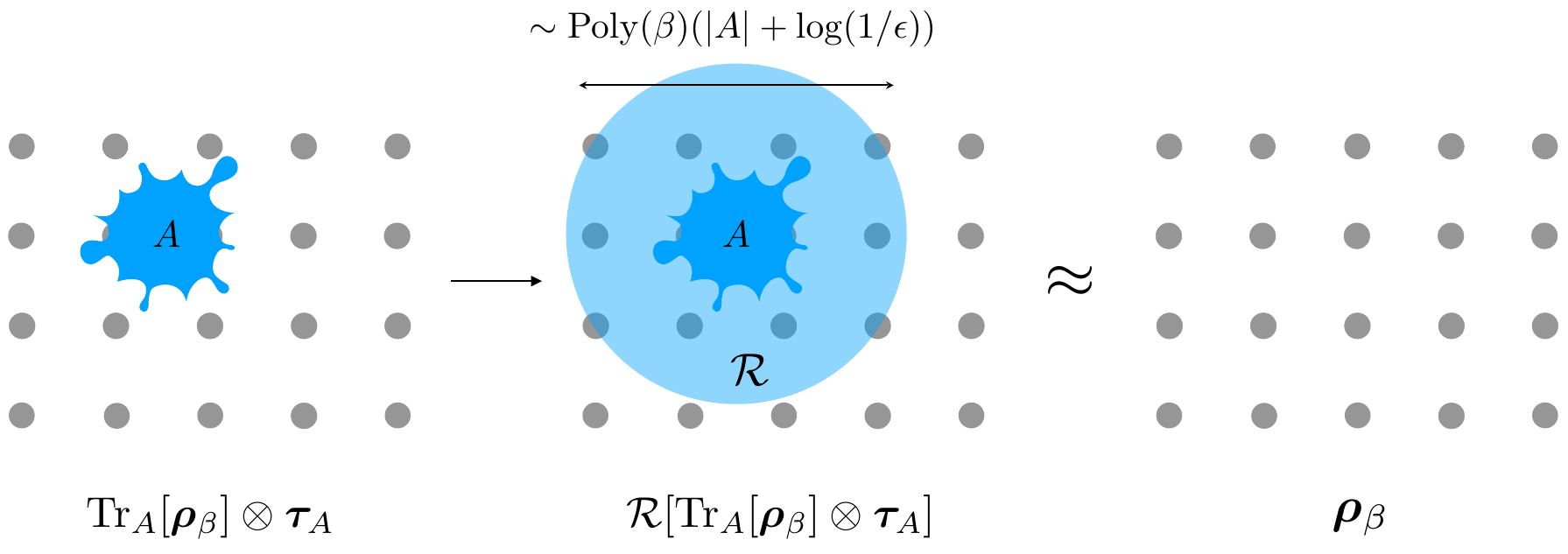}\caption{Our main result says that local disturbance to the Gibbs state can be recovered locally. Suppose we trace out a region $A$ and replace it with the maximally mixed state $\vec{\tau}_A$, the recovery map $\CR$ is quasi-local with radius growing with the region size $\labs{A}$ and the inverse temperature $\beta.$ In fact, the recovery map is a time-averaged detailed-balanced Lindbladian based on single-Pauli jumps on $A$.
}\label{fig:recovery_animation}
\end{figure}
\subsection{Efficient Gibbs sampling}

In the classical setting, Markovianity is a key tool for the construction and analysis of efficient sampling algorithms, when used in addition to the strong decay of other measures of correlations. Assume for instance the \textit{strong mixing condition}: for any regions $A\subset B\subseteq \Lambda$, and any site $x\in \Lambda\backslash B$, 
 \begin{align}\tag{Strong spatial mixing}\label{introdecaycorr}
\sup_\tau\,\left\|\rho^{B,\tau}_A-\rho^{B,\tau^x}_A\right\|_{\operatorname{TV}}\le C\,e^{-\operatorname{dist}(A,x)/\xi_s}
 \end{align}
 for some correlation length $\xi_s>0$ and constant $C$, where the supremum is taken over the set of all spin configurations $\tau\in \{-1,1\}^{\Lambda\backslash B}$, and where $\tau^x$ corresponds to the spin configuration equal to $\tau$ except on site $x$. Above, $\rho^{B,\tau}:=\e^{-\beta H^\tau_B}/\mathcal{Z}_B^\tau$ denotes the so-called conditional Gibbs distribution corresponding to the (classical) Hamiltonian $H^\tau_B:\{-1,1\}^{B}\to\mathbb{R}$ with spin configuration outside of $B$ fixed to $\tau\in \{-1,1\}^{\Lambda\backslash B}$, i.e.~$H^\tau_B(\delta_B)=\sum_{\gamma\cap B\ne 0}H_\gamma(\delta_{B\cap\gamma},\tau_{B^c\cap\gamma})$ for some interactions $H_\gamma:\{-1,1\}^\gamma\to[-1,1]$; $\rho^{B,\tau}_A$ denotes the marginal distribution on region $A$; finally, $\|.\|_{\operatorname{TV}}$ denotes the total variation distance. In other words, strong mixing requires that variations of the spins away from $B$ are not detected in the bulk $A$. 
 In the classical literature, this property is closely related to the uniqueness of the Gibbs distribution in the thermodynamic limit \cite{dobrushin1985completely,dobrushin1985constructive,dobrushin1987completely}. It was shown that \cite{martinelli1999lectures},
\begin{center}
\eqref{introeq.Markov}\quad +\quad \eqref{introdecaycorr}\quad $\Longrightarrow $\quad MCMC mixes in quasi-linear time.
\end{center}
 The main advantage of the above strategy compared to prior proofs~\cite{Aizenman1987,stroock1992equivalence,Majewski1995,temme2015fast,rouze2024optimal, rouze2024efficient} is that it provides physical, and often tight, static criteria for the efficient preparation of the Gibbs distribution.

In the quantum setting, how to systematically generalize the above paradigm remains open. So far, the most comparable approach of~\cite{brandao2019finite_prepare} requires an additional assumption of \textit{uniform Markov property} on top of a condition of \textit{uniform clustering of correlations} reminiscent from \eqref{introdecaycorr}: for any two non-overlapping regions $A,B\subset X\subset \Lambda$, 
 \begin{align}\tag{Decay correlations}\label{introcovariancedecay}
\|\vrho^X_{AB}-\vrho^X_{A}\otimes \vrho^X_B\|_1
\le \poly(\labs{A},\labs{B})\,e^{-\operatorname{dist}(A,B)/\xi_c}
 \end{align}
for some correlation length $\xi_c$ and constant $C'$, where $\vrho^X$
denotes the Gibbs state corresponding to the truncated Hamiltonian $\vH_X:=\sum_{\gamma\subseteq X}\vH_\gamma$. The authors constructed a quantum channel $\mathcal{A}$ composed of $\mathcal{O}(D)$ layers of log-local patches, leveraging the connection between Markov property and approximate local recovery channels \cite{fawzi2015quantum}. This channel satisfies $\|\mathcal{A}(\vrho) - \vrho_\beta\|_1 \leq \epsilon$ for any initial state $\vrho$. The channel $\mathcal{A}$ can be compiled into a nonexplict quantum circuit of size $\e^{\mathcal{O}(\log^D(n/\epsilon))}$.

Here instead, we exemplify the power of our approximate Markovianity by unconditionally removing the uniform Markov assumption, directly employing our time-averaged Lindblad dynamics \eqref{eq:RAt_intro} as recovery maps give $\vrho_\beta$ in $\e^{\mathcal{O}(\log^D(n/\epsilon))}$. Furthermore, assuming inverse polynomial control over the local gap of the algorithm's generators—an assumption valid for classical systems at sufficiently high temperatures \cite{martinelli1999lectures} and for commuting Hamiltonians under a strong clustering condition \cite{kastoryano2016commuting}—we can further reduce the mixing time to $\log^{\mathcal{O}(1)}(n/\epsilon)$, thereby achieving the regime of rapid mixing. 

\section{Preliminaries}

\subsection{The Lindbladian with exact detailed balance}
For any Hamiltonian $\vH$ on $n$ qubits, inverse temperature $\beta>0$ and a set of jumps $\{\vA^a\}_a$ containing their adjoints, each with norm $\norm{\vA^a}\le 1$, we consider the Linbladian~\cite{chen2023efficient}
\begin{align}
		\CL[\cdot] := \underset{\text{``coherent''}}{\underbrace{-\ri [\vB, \cdot]}} + \sum_{a}
		\int_{-\infty}^{\infty} \gamma(\omega) \left(\underset{\text{``transition''}}{\underbrace{\hat{\vA}^a(\omega)(\cdot)\hat{\vA}^{a}(\omega)^\dagg}} - \underset{\text{``decay''}}{\underbrace{\frac{1}{2}\{\hat{\vA}^{a}(\omega)^\dagg\hat{\vA}^a(\omega),\cdot\}}}\right)\rd\omega.\label{eq:exact_DB_L}
\end{align}
We recall the operator Fourier transform~\cite{chen2023quantum} of an operator $\vA$ associated to the Hamiltonian $\vH$ with spectral decomposition $\vH=\sum_{i}E_i\vP_{E_i}$, 

\begin{align}\label{eq:OFT}
{\hat{\vA}}(\omega)=  \frac{1}{\sqrt{2\pi}}\int_{-\infty}^{\infty} \e^{\ri \vH t} \vA \e^{-\ri \vH t} \e^{-\ri \omega t} f(t)\rd t =\sum_{\nu\in B(\vH)} \vA_{\nu}\hat{f}(\omega - \nu)
    \end{align}
where the Bohr frequencies $\nu\in B(\vH)$ are the set of energy differences, and  $\vA_\nu:=\sum_{E_2-E_1=\nu}\vP_{E_2}\vA \vP_{E_1}$ are eigenoperators of Heisenberg evolution, with a Gaussian weight with an energy width $\sigma>0$
    \begin{align}
        \hat{f}(\omega)=\frac{1}{\sqrt{{\sigma}\sqrt{2\pi}}} \exp\L(- \frac{\omega^2}{4\sigma^2}\R),\quad \text{and}\quad f(t) = \e^{-\sigma^2t^2}\sqrt{\sigma\sqrt{2/\pi}}.\label{eq:fwft}
    \end{align}
Recall the Fourier transform pairs
    \begin{align}
        \hat{f}(\omega)=\frac{1}{\sqrt{2\pi}}\int_{-\infty}^{\infty}\e^{-\ri\omega t} f(t)\mathrm{d}t\quad \text{and}\quad f(t)=\frac{1}{\sqrt{2\pi}}\int_{-\infty}^{\infty}\e^{\ri\omega t} \hat{f}(\omega)\mathrm{d}\omega. 
    \end{align}
    We will mainly consider the Metropolis weight
\begin{align}
    \gamma(\omega) = \exp\L(-\beta\max\left(\omega +\frac{\beta \sigma^2}{2},0\right)\R)\label{eq:Metropolis},
\end{align}
but the Gaussian transition weight sometimes guides the computation
\begin{align}
\gamma^G(\omega) &= \exp\L(- \frac{(\omega + \omega_{\gamma})^2}{2\sigma_{\gamma}^2}\R)\quad \text{with variance}\quad \sigma_{\gamma}^2 := \frac{2\omega_{\gamma}}{\beta}-\sigma^2. \label{eq:Gaussianweight}
\end{align}
In our final bounds, we will choose energy width $\sigma=\frac{1}{\beta}$\footnote{This $\sigma$ was sometimes denoted with a subscript $\sigma_E$.}; but for transparency, we keep $\sigma$ a tunable parameter in our lemmas. The parameters $\omega_{\gamma}$ and $\sigma_{\gamma}$ are tunable subject to the constraint $\beta(\sigma_{\gamma}^2+\sigma^2) := 2\omega_{\gamma}$. For now, we do not need to worry about the explicit form of $\vB$; what will matter is the Dirichlet form (\autoref{sec:Dirichlet}), where the terms are rearranged and properly conjugated by the Gibbs state. The jumps $\vA$ we use will be Paulis chosen from a subset $P$ of the set $P_{[n]}$ of all $n$-qubit Pauli strings (cardinality denoted $|P|$, each normalized by $\norm{\vA} = 1$), but some of our lemmas are stated with more generality.

We recall that, given a full-rank state $\vrho$, a Lindbladian $\CL$ is said to satisfy the Kubo-Martin-Schwinger (KMS)-$\vrho$-detailed balance if it is symmetric with respect to the KMS inner product associated with $\vrho$
\begin{align}
\langle \vX,\vY\rangle_{\vrho}:=\tr[\vX^\dagger \vrho^{\frac{1}{2}}\vY\vrho^{\frac{1}{2}}]\,.
\end{align}
We denote by $\norm{\vX}_{\vrho} := \sqrt{\langle \vX,\vX\rangle_{\vrho}}$ the $\vrho$-weighted norm induced by the KMS inner product. 

\begin{rmk}
Other $\vrho$-weighted norms are possible (e.g., GNS~\cite{kastoryano2016commuting}), but in this paper, we will only consider the KMS inner product since the Lindbladian we consider is KMS-detailed balanced.    
\end{rmk}
The conversion to operator norm always holds, but sometimes may be suboptimal.
\begin{lem}[Operator norm controls weighted norms and inner-product]\label{lem:operatornorm}
    Unconditionally, we have that $\norm{\vX}_{\vrho}\le \norm{\vX}$ and $\braket{\vX,\vY}_{\vrho} \le \norm{\vX}\norm{\vY}.$
\end{lem}

Next, we recall some properties of the generators introduced in \cite{chen2023efficient}:

\begin{thm}[{\cite{chen2023efficient}}]
The Lindbladian $\CL$ defined in Equation \eqref{eq:exact_DB_L} satisfies KMS-$\vrho_\beta$-detailed balance and hence fixes the Gibbs state exactly:
\begin{align}
    \CL[\vrho_\beta] =0\quad \text{where}\quad \vrho_\beta \propto \e^{-\beta \vH}. 
\end{align}

\end{thm}

\begin{thm}[{\cite{chen2023efficient,chen2023quantum}}]\label{circuitimplementation}
Consider a set of jumps $\norm{\vA^a}\le 1$ with cardinality $\labs{P}$. Then, the time evolution for the Lindbladian $\CL$ can be simulated in $\epsilon$-diamond distance with costs\footnote{Note the difference in time scaling compared to \cite[Theorem I.1]{chen2023efficient}, which is due to the difference in normalization of the jumps $\vA^a$. Here, each jump has operator norm one.}:
\begin{itemize}
\item[] $\tilde{\mathcal{O}}(|P|t\beta)$ total Hamiltonian simulation time;
\item[] $\tilde{\mathcal{O}}(1)$ resettable ancilla;
\item[] $\tilde{\mathcal{O}}(|P|t)$ block-encodings for the jumps $\frac{1}{\sqrt{|P|}}\sum_{a\in P}|a\rangle\otimes \vA^a$;
\item[] $\tilde{\mathcal{O}}(|P|t)$ other two-qubit gates.
\end{itemize}
The $\widetilde{O}(\cdot)$ notation absorbs logarithmic dependencies on $n$, $t$, $\beta$, $\norm{\vH}$, $1/\epsilon$ and $|P|$. 
\end{thm}
\begin{rmk}
    The same dependencies hold for the time average map $\mathcal{R}_{A,t}[\cdot]:=\frac{1}{t}\int_0^t\e^{s\mathcal{L}}[\cdot]\rd s$ by simply sampling some random time $s$ uniformly at random over $[0,t]$ and running the evolution generated by $\mathcal{L}$ up to time $s$.
\end{rmk}

\subsection{Hamiltonian with bounded interaction degree}\label{sec:Ham}
On a set $\Lambda$ of $n = \labs{\Lambda}$ qubits, we consider Hamiltonians $\vH$ with few-body terms $\vH_{\gamma}$ 
\begin{align}
    \vH = \sum_{\gamma\in \Gamma} \vH_{\gamma}\quad \text{where}\quad \norm{\vH_{\gamma}} \le 1.
\end{align}
From this decomposition, we define the interaction graph with vertices corresponding to the set $\Gamma$, and we draw an edge between $\gamma_1$ and $\gamma_2$ if and only if the terms have overlapping supports (self-looping allowed):
\begin{align}
\gamma_1\sim \gamma_2\quad \iff\quad    \text{Supp}(\vH_{\gamma_1}) \cap \text{Supp}(\vH_{\gamma_2}) \ne \emptyset.
\end{align}
Similarly, we may consider any subset of vertices $A\subset \Lambda$ and write 
\begin{align}
    A \sim \gamma \iff A\cap \text{Supp}(\vH_{\gamma})\ne \emptyset.
\end{align}
The maximal degree of the interaction graph is denoted by $d$, and we are particularly working in the regime where $d$ is a constant independent of the system size $n$; this will ensure the possibility of conjugating by a constant temperature Gibbs state (see \autoref{lem:convergence_imaginary}). 
For any two subsets of vertices $A,B\subset \Lambda$, we denote by $\operatorname{dist}(A,B)$ the minimal length of a path connecting $A$ to $B$ via interactions in $H$:
\begin{align*}
\operatorname{dist}(A,B)=\min\bigg\{\ell\in\mathbb{N}:\exists \gamma_1,\dots\gamma_{\ell}\in\Gamma\quad \text{such that}\quad A\sim \gamma_1\sim\gamma_2\sim\dots \sim \gamma_{\ell}\sim B\bigg\}\,.
\end{align*}
Often, we will also consider the subset $A$ or $B$ to the supports $\operatorname{Supp}(\vH_\gamma)$ of Hamiltonian term $\vH_\gamma$, and we will simply abuse the notation to write  $\operatorname{dist}(\gamma,\gamma')$ and $\operatorname{dist}(A,\gamma')$.

For later parts of our arguments, for a region $A\subset \Lambda$, we often consider the local Hamiltonian patch $\vH_{\ell}$ containing all terms $\vH_{\gamma}$ with distance at most $\ell-2$ from $A$. 
\begin{align}
    \vH_{\ell} := \sum_{\gamma:\operatorname{dist}(\gamma,A) <  \ell-1} \vH_{\gamma}.
\end{align}

\begin{rmk}
The system size $n$ does not feature in our arguments, and we believe that the same could be formalized for infinitely large systems.
\end{rmk}

\section{Main results}\label{sec:main results}

From here onward, we denote by $\vS \in P_A$ the set of all non-trivial Pauli strings $\vS$ on $A$ (excluding the identity string), and by $\vA^a \in P_A^1$ the subset of single-qubit Pauli matrices, which will be the jumps of our Lindbladian. The main result states that if we discard a region $A$ of the Gibbs state
\begin{align}
 \vrho_\beta \rightarrow \tr_{A}[\vrho_\beta]\otimes \vec{\tau}_A=:\vrho_{_\beta,-A},
\end{align}
then, running a Gibbs sampler with jumps on $A$ for a long enough time recovers the Gibbs state. Here, $\vec{\tau}_A$ denotes the maximally mixed state on $A$. Given a region $A\subset \Lambda$ and a tunable time parameter $t\ge0$, consider the quantum channel (Completely Positive and Trace-preserving, CPTP map)
\begin{align}
    \CR_{A,t}[\cdot] &:= \frac{1}{t}\int_{0}^t \exp\L(s\,\mathcal{L}_A\R)[\cdot] \,\rd s \\
    \mathcal{L}_A &:=\sum_{a\in P^1_A}\CL_a,\quad\text{where}\quad P_A^1 :=\{\vX_i,\vY_i,\vZ_i\}_{i\in A}.  \label{eq:RAt}
\end{align}
Each $\CL_a$ is the Lindbladian associated with each single-qubit Pauli jump $\vA^a$ (with metropolis weight). 

\begin{thm}[Quasi-local recovery maps via time-averaged Gibbs sampling] \label{thm:main}
Consider the Gibbs state of a Hamiltonian $\vH$ with interaction degree at most $d$ and a region $A\subset \Lambda$. Then, the time-averaged Lindblad dynamics $\CR_{A,t}$ with single-qubit Pauli jumps $\{\vA^a\}_{a\in P^1_A}$, Metropolis weight $\gamma(\omega)$~\eqref{eq:Metropolis}, and $\sigma = 1/\beta$, $t>0$ gives an approximate recovery map at all temperatures $\beta$ 
    \begin{align}
        \norm{\CR_{A,t} [\vrho_{\beta,-A}] - \vrho_{\beta}}_1 \le |A|^22^{2|A|} \begin{cases} r(\beta,d) \cdot   t^{-\frac{128\beta^4_0}{\beta^3(\beta+5\beta_0)}}& \text{if}\quad \beta>4\beta_0,\\
        r'(\beta,d)  \cdot t^{-\frac{2\beta_0}{\beta+5\beta_0}} & \text{if}\quad \beta\le 4\beta_0,
        \end{cases}
\end{align}
for some explicit functions $r(\beta,d)$ and $r'(\beta,d)$, where $\beta_0 := 1/4d$. Therefore, there are numbers $r,\mu >0$ and $0<\lambda<1$ depending only on $\beta,d$ such that 
 \begin{align}
        \norm{\CR_{A,t} [\vrho_{\beta,-A}] - \vrho_{\beta}}_1 \le  r \e^{\mu|A|}\,t^{-\lambda}\label{eq:mu_lambda}.
\end{align}
\end{thm}
See~\autoref{sec:proofmainthm} for the proofs.
\begin{rmk}
The above remains to hold for $\beta_0 =(1-\epsilon)/2d$ for any fixed $\epsilon>0$. However, at $\beta_0 = 1/2d,$ $\norm{\vrho_{\beta_0}\vA\vrho_{\beta_0}^{-1}}$ may grow with the system size $n$ and introduce extra $n$-dependence on the RHS.
\end{rmk}
\begin{rmk}
The exponential dependence on $\labs{A}$ is hard to remove unconditionally using the current Lindbladian approach. Right now, it appears due to the slow, inverse polynomial decay of the Dirichlet form. If the present argument can be combined with a faster mixing time or spectral gap analysis, one might be able to improve the exponential dependence on $\labs{A}$, hence establishing the global Markov property, see \autoref{sec:improvegap}.
\end{rmk}

The recovery map can be localized using standard Lieb-Robinson bounds for Hamiltonian with bounded interaction degree (\autoref{lem:truncation}). See~\autoref{sec:pf_qlocal_decay} for the proof. Here and throughout the paper, we write
\begin{align}
a \lesssim b \quad \text{iff}\quad a \le c b \quad \text{for an absolute constant} \quad c>0.
\end{align} 
\begin{cor}[Quasi-locality estimates]\label{cor:main_qlocal}
    For a region $A\subset \Lambda$, the approximate recovery map $\CR_{A,t}$ can be well-approximated by a strictly local map $\CR_{A,t,\ell} $ supported on qubits at distance at most $\ell$ from $A$:
    \begin{align}
        \norm{\CR_{A,t,\ell} - \CR_{A,t}}_{1-1} \lesssim |A|t (\e^{-c'\frac{\ell}{d\beta}} + 2^{-\ell})
    \end{align} for some absolute constant $c'$.
    Therefore, there is a time $t^*(\ell) = \e^{(\mu \labs{A}+m\ell)/{\lambda+1}}$ such that 
    \begin{align}
        \lnorm{\vrho_\beta-\mathcal{R}_{A,t^*,\ell}[\vrho_{{\beta,-A}}]}_1 \lesssim r \exp\L(\frac{\mu'\labs{A} -m\lambda \ell}{2} \R)
    \end{align}
    where $m=\min\L(\ln(2),\frac{c'}{d\beta}\R)$, $\mu' = \mu+2$, and $r, \mu, \lambda $ as in~\eqref{eq:mu_lambda}.
\end{cor}
\begin{rmk}
We expect similar estimates to follow for Gibbs states over Fermionic systems in the even parity sector, since the tools used to show our bounds, such as Lieb-Robinson bounds or expressing partial traces as localized random unitary channels, directly extend to this setup, see, e.g.~\cite{nachtergaele2018lieb,haah2020quantum}.
\end{rmk}

\subsection{Decay of Conditional mutual information for tripartitions}
The decay of QCMI implies the following approximate Markov property by standard entropic continuity bounds. Recall, the conditional mutual information of a tripartite state $\vrho_{ABC}$ is defined as 
\begin{align*}
I(A:C|B)_{\vrho}:=S(\vrho_{AB})+S(\vrho_{BC})-S(\vrho_B)-S(\vrho_{ABC})\,,
\end{align*}
with the entropy of a state $\vsigma$ denoted by $S(\vsigma):=-\tr(\vsigma\log\vsigma)$.

\begin{cor}[Quantum Gibbs states are locally Markov] \label{cor:CMI}
Consider a tripartition $\Lambda = ABC$ with region $A\subset \Lambda$ shielded by $B$. Then, if $ \operatorname{dist}(A,C) \ge 4\e^2\beta d$, the conditional mutual information satisfies
\begin{align}
    I(A:C|B)_{\vrho_\beta}\lesssim r' |A||C|\exp\bigg(\mu' \min(\labs{A},\labs{C}) -\lambda'\operatorname{dist}(A,C)\bigg)\,,
\end{align}
for some numbers $r'$, $\mu'$ and $\lambda'$ (as in~\eqref{eq:mu_lambda}) which only depend on the inverse temperature $\beta$ and the degree $d$ of the interaction graph.
\end{cor}
See~\autoref{sec:pf_qlocal_decay} for the proof. Our arguments focus on tripartitions $ABC\subset \Lambda$; the current argument does not handle the case with more refined partitions $ABCD$~\cite{kuwahara2020clustering}. 

\begin{rmk}
The prefactor in the above corollary scales exponentially with the size of the smaller region and linearly with the larger region. In comparison, the quasi-local recovery statement (\autoref{cor:main_qlocal}) does not refer to the global system size and can operate in the thermodynamic limit; this loss is due to $\log(dim)$ factors common in conversion between entropies and trace distances.\end{rmk}

\subsection{Quasi-local preparation algorithm assuming uniform clustering}
Another application of \autoref{thm:main} is the following guarantee for the preparation of Gibbs states on $D$-dimensional hypercubic lattices. Recall, ref.~\cite{brandao2019finite_prepare} gave a quasi-local preparation algorithm for quantum Gibbs states under a \textit{uniform clustering} and a \textit{uniform Markov} condition. As a demonstration, we use the newly proven unconditional local Markov property (\autoref{thm:main}) to get rid of the second condition. In this section, we assume that the Hamiltonian is of \textit{finite range}, meaning that the non-zero interactions are localized on regions of finite diameter with respect to the standard lattice distance. We first introduce some notation. For any subset $X\subset \Lambda_L$ of the lattice $\Lambda_L=[-L,L]^{D}$ with $|\Lambda_L|=n$, we write the truncated Gibbs state by
\begin{align}
    \vrho_\beta^X := \e^{-\beta \vH_X}/\tr[\e^{-\beta \vH_X}]\quad \text{where}\quad \vH_X=\sum_{Z\subseteq X}\vh_Z  \,.
\end{align}
Here, we assume that the interactions $\vh_Z$ with $\|\vh_Z\|\le 1$ are supported on regions $Z$ of the lattice $\Lambda_L$ such that for any region $Z$ of diameter larger than a constant $r$, $\vh_Z=0$. For any state $\vsigma$ and pair of observables $\vA,\vB$, we define a covariance
\begin{align}
    \operatorname{Cov}_{\vsigma}(\vA,\vB) := \labs{\tr[\vsigma\vA\vB]-\tr[\vsigma\vA]\tr[\vsigma\vB]}.
\end{align}

\begin{defn}[Uniform clustering]\label{defn:unifclust}
Consider a Hamiltonian $\vH$ with an interaction graph and an inverse temperature $\beta$. We say the pair $(\vH,\beta)$ is uniformly clustering if for any regions $A,C\subseteq X\subseteq \Lambda$ such that $\operatorname{dist}(A,C)\ge \ell$, we have
\begin{align}
\operatorname{Cov}_{\vrho_{\beta}^X}(\vA,\vC)\le \poly(|A|\,| C|) \cdot \norm{\vA}\norm{\vC} \e^{-\frac{\ell}{\xi}}
\end{align}
for some correlation length $\xi>0$ and any operators $\vA$ supported on $A$ and $\vC$ supported on $C$. 
\end{defn}
\begin{rmk}
    The above definition is relaxed slightly from~\cite{brandao2019finite_prepare} by allowing polynomial prefactors of the volumes $\labs{A},\labs{C}$ that may arise in tools to prove uniform clustering (e.g., liberal use of Lieb-Robinson bounds). The tighter scaling of $\min(\labs{\partial A},\labs{\partial C})$ was proven in~\cite{Kliesch19} in a high-temperature regime. A version with local observable and at superlogarithmic distances $\Omega(\log(n))$ was proven in~\cite{Harrow20} under the assumption of certain complex analytic properties of the free energy.
\end{rmk}
Roughly, uniform clustering demands that far-apart observables have a decay of covariance under any restricted Gibbs state for any subset $X$ (which may be topologically nontrivial). For our purposes, the key intermediate consequence of uniform clustering is \textit{local indistinguishability}. As shown in~\cite{brandao2019finite_prepare}, one can prepare the Gibbs states by stitching quasi-local patches together.

\begin{thm}[uniform clustering implies local indistinguishability{~\cite[Theorem 5]{brandao2019finite_prepare}}]
\label{thm:localindistinguishability}
Consider a Hamiltonian $\vH$ on a $D$-dimensional lattice and an inverse temperature $\beta$. Suppose the pair $(\vH,\beta)$ is uniformly clustering with correlation length $\xi$ (\autoref{defn:unifclust}), then, the pair satisfies local indistinguishability: For any $ABC=X\subset \Lambda,$ with the distance $\operatorname{dist}(A,C) = \ell,$ we have that
\begin{align}
\norm{\tr_{BC}[\vrho^X] - \tr_{B}[\vrho^{AB}] }_1\le {\e^{c'\beta}} \labs{\partial_{B} C}\L(\poly(\labs{A},\ell^{D})\e^{-\ell/{2}\xi} +  \e^{-\ell/c}\R)
\end{align}
for some universal constant $c$ and a constant $c'>0$ which depends on $\beta$ and the locality of $H$ (see \cite{kim2012perturbative,Kato2019,brandao2019finite_prepare} for further details), and where $\partial_B C$ is the boundary of $C$ with $B$.
\end{thm}
\begin{cor}[Uniform clustering implies quasi-local preparation]\label{cor:samplingresult}
Under the assumption of uniform clustering \autoref{defn:unifclust} for $\vH$ at inverse temperature $\beta$, there exists a channel  with (dissipative) gate complexity $\e^{\mathcal{O}(\log^{D}(n/\epsilon))}$ outputting a state $\vrho'$ such that $\|\vrho'-\vrho_\beta\|_1\le \epsilon$.
\end{cor}

\begin{rmk}
The runtime of the algorithm in \autoref{cor:samplingresult} can be improved to a (quasi-)optimal scaling of $\log^{\mathcal{O}(1)}(n/\epsilon)$ under the further condition that the Lindbladians $\mathcal{L}_{A,\ell}$ defined on every region $A$ have a gap that at least inverse polynomial in the size $|A|$, {see \autoref{sec:improvegap}}. Right now, the local condition we imposed might be overly stringent, but we wanted to illustrate the possibility of improving the argument assuming certain local mixing as input. 
\end{rmk}

\section{key ideas}
In this section, we outline the key steps of our approach, highlighting the novel aspects, while deferring further details to the following sections (see \autoref{sec:proofmainthm} for the overall proof). The proof intuition behind \autoref{thm:main} is that the long-time averaging map $\CR_{A,t}$, which keeps updating region $A$ using all supported single qubit Pauli jumps, must mix $A$ very thoroughly, ``conditioned'' on everything else. In particular, we will work in the Heisenberg picture by making use of the duality
\begin{align}\label{rhotoX}
    \norm{\vrho_\beta-\CR_{A,t}[\vrho_{\beta,-A}]}_1 = \sup_{\norm{\vX}\le 1} \labs{\tr[\vX(\vrho_\beta-\CR_{A,t}[\vrho_{\beta,-A}])]} = \sup_{\norm{\vX}\le 1} \labs{\tr[(\CR^{\dagger}_{A,t}[\vX] - (\CR^{\dagger}_{A,t}[\vX])_{-A}) \vrho_\beta]}
\end{align}
where the second equality uses the fixed point property $\CR_{A,t}[\vrho_\beta]=\vrho_\beta$. Here, we denote by $\CR^\dagger$ the adjoint of a quantum channel $\CR$ w.r.t. the Hilbert-Schmidt inner product. Therefore, the statement of recoverability of $A$ is reduced to showing that the evolved operator is nearly trivial on region $A$ 
\begin{align}
\CR^{\dagger}_{A,t}[\vX] \approx ({\CR^{\dagger}_{A,t}[\vX])_{-A}}.
\end{align}
However, quantifying the above approximation requires care, as we do not make any structural assumption on the state $\vrho_\beta$. We begin with reducing trace distance in the objective function to bounding the following nested commutator with Pauli strings on region $A$
\begin{align}\label{doublecommutatortoDirichlet}
  \frac{1}{4^{\labs{A}}}\sum_{\vS\in P_A}   \norm{[\vS,[\vS,\CR^{\dagger}_{A,t}[\vX]]]}_{\vrho_\beta}\,\quad \text{for any}\quad \vX.
\end{align}

\subsection{A H\"{o}lder-like inequality in the weighted norm}

The first ingredient is a Holder-like inequality for products in the weighted norm. It allows us to ``peel off'' the outer commutator by 
\begin{align}
\norm{\vS\vO}_{\vrho_\beta}\le 
    r \e^{\mu\labs{\Supp(\vS)}} \cdot \norm{\vO}_{\vrho_\beta}^{\nu} \quad \text{for any operator $\vO$ and Pauli string $\vS$}
\end{align}
for any inverse temperature $\beta$, with parameters $r,\mu,\nu>0$ depending only on $\beta,d$ (see~\autoref{lem:weight_norm_comm} and \autoref{cor:highweight}). The main novelty here is apparent when compared with a naive application of H\"{o}lder's inequality 
\begin{align}
    \norm{\vS\vO}_{\vrho_{\beta}} \le \norm{\vrho^{1/4}_{\beta}\vS\vrho^{-1/4}_{\beta}}\cdot  \norm{\vO}_{\vrho_{\beta}} \qquad \qquad\text{(Naive H\"{o}lder)},
\end{align}
which depends on a conjugation with the Gibbs state, the norm of which may generally diverge at low temperatures 
 $   \norm{\vrho_\beta \vA \vrho_\beta^{-1}} \sim \e^{\Omega(n)}$ \cite{perez2023locality}.
In fact, this is a genuinely noncommutative phenomenon absent in commuting or classical Hamiltonians. Roughly, a local operator $\vA$ can change the energy by a lot with an exponentially small amplitude; once $\beta$ gets large enough, it sufficiently amplifies the exponentially small amplitudes, causing a divergence. 

Remarkably, in our proofs, we found a systematic and conceptually transparent way to regularize this divergence, by decomposing the operator by an operator Fourier transform (\autoref{lem:sumoverenergies})
\begin{align}
    \vA =  \frac{1}{\sqrt{\sigma2\sqrt{2\pi}}} \int_{-\infty}^{\infty} \hat{\vA}(\omega)\rd \omega = \frac{1}{\sqrt{\sigma2\sqrt{2\pi}}} \L(\int_{\labs{\omega}\le \Omega}+\int_{\labs{\omega}>\Omega}\R) \hat{\vA}(\omega)\rd \omega.
\end{align}
The operator Fourier transform selects matrix elements that change the energy by roughly $\omega$, which allows us to control the effect of Gibbs conjugation (\autoref{lem:bounds_imaginary_conjugation})
\begin{align}
\norm{\e^{\beta \vH} \hat{\vA}(\omega) \e^{-\beta \vH}} \le \e^{\beta\omega}\cdot \frac{\e^{\sigma^2\beta^2}}{\sqrt{{\sigma}\sqrt{2\pi}}}  \norm{\vA}.  
\end{align}
Proper choices of the truncation frequency $\Omega$ allow us to avoid the divergences. Technically, the above implicitly exploits the Gaussian uncertainty $\hat{f}(\nu)$ of the operator Fourier transform as a Gaussian decay with a \textit{fixed} variance always dominates the exponential growth at \textit{any} exponent (See~\autoref{fig:regularize}).

We may further reduce commutators with Pauli strings $\vS\in P_A$ to commutators with single qubit Pauli jumps $\vA^a$ (\autoref{cor:global_local_comm}) by another use of the aforementioned Gibbs conjugation regularization bounds derived in \autoref{lem:weight_norm_comm} and \autoref{cor:highweight}: for any inverse temperature $\beta>0$, and $\vS=\prod_{a=1}^w\vA^a\in P_A$,
\begin{align}
    \norm{[\vS,\vO]}_{\vrho_\beta} \le 2^{w} r'\sum_{a} \lnorm{[\vA^a,\vO]}_{\vrho_\beta}^{\nu'}
\end{align}    
for some constants $r',\nu'$ depending on $\beta$. Therefore, we may focus entirely on the bounding weighted norm of the commutator with single Paulis.

\begin{figure}[t]
\includegraphics[width=0.9\textwidth]{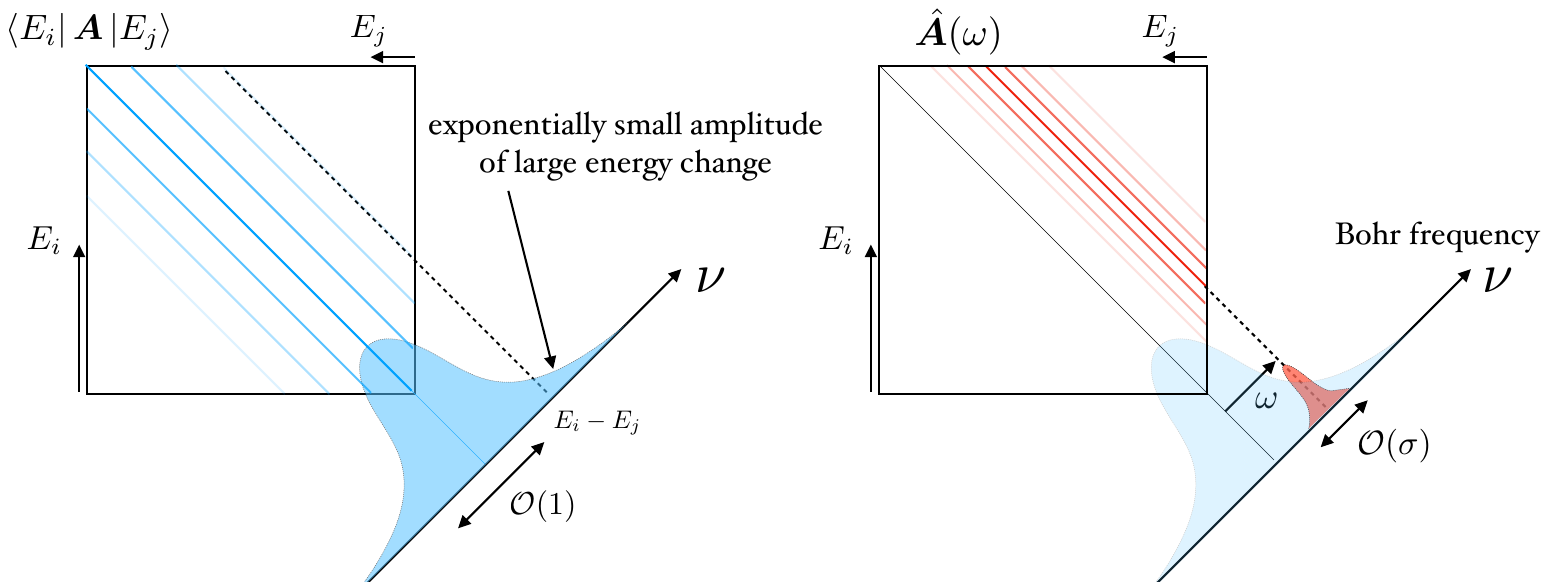}\caption{ Any operator can be decomposed by the Bohr frequencies $\vA = \sum_{\nu} \vA_{\nu}$. When the operator $\vA$ acts on a single site, the amplitudes concentrate around $\nu = \CO(1) \sim \norm{[\vH,\vA]}.$ However, when the Hamiltonian is non-commuting and beyond one dimensions, there could be an exponentially small amplitude for large energy changes $\nu$, which causes divergence for the imaginary time conjugation $e^{\beta \vH}\vA e^{-\beta \vH}$ at a large constant $\beta$. The operator Fourier transform $\hat{\vA}(\omega)$ with Gaussian weights selects the amplitudes near $\omega\pm \CO(\sigma)$. The Gaussian tail is particularly effective for mitigating the exponential divergence due to imaginary time conjugation.
}\label{fig:regularize}
\end{figure}
\subsection{Controlling commutator norms by Dirichlet forms}
Another ingredient in our proof is the link between Dirichlet forms and commutators (\autoref{lem:commutator_Dirichlet}). 
Recall that the Lindbladian $\CL_{A}$ is associated with the Dirichlet form
\begin{align}
    \CE_{A}(\vX) = - \langle \vX,\,\CL^\dagger_A(\vX)\rangle_{\vrho_\beta} =- \sum_{a\in P^1_A} \langle \vX,\, \CL^\dagger_a(\vX)\rangle_{\vrho_\beta} =  \sum_{a\in P^1_A} \CE_a(\vX).
\end{align}
We can show that, for any single-qubit Pauli jump $\vA^a$ and all bounded observables $\vX$, $\|\vX\|\le 1$, 
\begin{align}
    \lnorm{[\vA^a,\vX]}_{\vrho_{\beta}} &\le r''\,
    \CE_{a}(\vX) ^{\nu''}
    \end{align}    
for some $r'',\mu'',\nu''>0$ depending only on $\beta,d$. 

The backbone of this inequality is the following exact expression (\autoref{lem:Dirichlet_Gaussianmetropolis}): for any $a\in P^1_A$,
\begin{align}
        \CE_a(\vX) =\int_{-\infty}^{\infty}\int_{-\infty}^{\infty} g(t) h(\omega)\|[\hat{\vA}^a(\omega,t),\vX] \|_{\vrho_\beta}^2\rd t \,\rd \omega,\label{DirichletformAnice}
\end{align}
where $\hat{\vA}^a(\omega,t):=\e^{i\vH t}\hat{\vA}^a 
(\omega)\e^{-i\vH t}$, and for some positive functions $g, h> 0$. Remarkably, this yields an elegant, manifestly PSD quadratic form of commutators. Given the utility of analogous expressions for Dirichlet forms in the classical literature, we expect that \eqref{commutantrule} will also be beneficial in other contexts.

Consequently, an exact stationary operator must satisfy 
\begin{align}\label{commutantrule}
    \CL^{\dagger}_a[\vX] = 0 \iff \CE_a(\vX,\vX) =0 \iff  [\hat{\vA}^a(\omega,t),\vX ] = 0 \quad \text{for almost all}\quad \omega, t.
\end{align}
In particular, by taking linear combination (\autoref{lem:sumoverenergies}), considering all possible single-body jumps, and sending $t\rightarrow 0$,
\begin{equation}\label{dirichletkernel}
    \Big[\vX, \int_{-\infty}^{\infty}\hat{\vA}^a(\omega)\rd \omega\Big] \propto [\vX,\vA^a]=0\,\quad \text{for all}\quad a\in P^1_A.
\end{equation}

Therefore, by Leibniz rule for commutators, for any Pauli string $\vS=\prod_{a}\vA^a\in P_A$, $[\vX,\vS]=0$, which implies that the kernel of the function of $\vX$ in \autoref{doublecommutatortoDirichlet} is included in that of $\CE_A$, which is consistent with our quantitative bound.

\subsection{Polynomial decay of Dirichlet form}

It remains to show that the Dirichlet form associated with the single-Pauli jumps $\vA^a\in P^1_A$ decays for $t\to\infty$. Remarkably, the time-average map has unconditional polynomial decay of the Dirichlet form under (\autoref{cor:dirichlet}):
\begin{align}
    \CE_A(\CR^{\dagger}_{A,t}[\vX]) \le \frac{2}{t}.
\end{align}
Note that such a property is generally false for the Lindblad evolution $\e^{\CL_A^{\dagger} t}$ itself without time-averaging, because $\CL_A^{\dagger}$ may have arbitrarily small eigenvalues. Thus, time-averaging provides a different mechanism to obtain a small Dirichlet form that is independent of the spectral gap. Indeed, imposing an additional local gap condition implies an exponentially faster decay of the Dirichlet form and would significantly improve our Gibbs sampling results (see \autoref{sec:improvegap}). Roughly, this means that an operator changing slowly under the Lindbladian must nearly commute with the jump $\vA^a$. The claim follows after combining all the aforementioned bounds:
\begin{align*} 
\|\CR^{\dagger}_{A,t}[\vX] - ({\CR^{\dagger}_{A,t}[\vX])_{-A}}\|\le \e^{\mu |A|}\CE_A(\CR^{\dagger}_{A,t}[\vX])^{\nu}\le r\frac{\e^{\mu |A|}}{t^\nu}\,.
\end{align*}
\section{Discussions and outlook}

We have proved the local Markov property of quantum Gibbs states for any Hamiltonian with a bounded interaction degree at any constant temperature. Tracing out region $A\subset \Lambda$ of a Gibbs state, there is a recovery map approximately localized around $A$ that recovers the Gibbs state. Remarkably, this static property is proven using the dynamics: the recovery map is a time-averaged Lindblad dynamics with single-Pauli jumps on $A$. Consequently, the conditional mutual information for tripartitions $ABC = \Lambda$, where $B$ shields $A$ from the remaining sites, decays exponentially with the shielding distance. However, the bound on CMI grows exponentially with the size $\labs{A}$, which comes from the possibility of an exponentially long mixing time. Nevertheless, this local Markov property is already sufficient for Gibbs state preparation using quasi-local patches, assuming uniform clustering of covariance. If we further assume that the local gap of the Lindbladians with jumps on region $A$ decays polynomially in the $\labs{A}$, we may improve the CMI bound. Still, for general Hamiltonians at low temperatures, the global Markov property remains open.  

Our proof introduces a family of new analytic toolkit that adds to the Gibbs sampling literature. To handle conjugation with a low-temperature Gibbs state $\e^{\beta \vH} \vA \e^{-\beta\vH}$, we introduce a handy decomposition into operator Fourier transforms $\hat{\vA}(\omega)$. By properly truncating the frequency tail, we regularize the divergences at low temperatures. To relate the dynamics to the statics, we give an explicit commutator-square expression for the Dirichlet form. We believe the ingredients will be cornerstones to the holy grail of rapid mixing times from the decay of correlation for quantum Gibbs samplers. 

It is intriguing to contrast our Gibbs state results with those of gapped ground states. Our local Markov property extends smoothly at low temperatures (length scales growing as $\poly(\beta)$). Thus, if we impose conditions on the density of states near ground states~\cite{hastings2007entropy}, our results can be relevant to gapped-ground state properties as soon as $\beta \gtrsim \log(n)/\Delta.$ A closely related result is that gapped ground states always enjoy decay of correlation~\cite{hastings2006spectral}, which can be converted to mutual information decay (with a loss dependent on the region size). Since for pure states, the mutual information and conditional mutual information are strictly equal, our bounds as a black box are not new. Still, the fact that the recovery map can be taken to be a thermalization dynamics, and the fact that the Lindbladian depends largely on the nearby Hamiltonian, is new. Recently, the ``entanglement-bootstrap'' program~\cite{shi2020fusion,kim2024learning,kim2025Topological} has provided a new understanding of the structure of gapped ground states with additional assumptions on multipartite correlation beyond the CMI. It would be interesting to ask whether such approaches for mixed states could demystify thermal correlation beyond the Markov properties we studied.

\section{Road map}
We begin with a recap of notations, followed by
sections organized by the key ingredients in our proofs. First, we instantiate the elementary but useful properties of the time-averaging map (\autoref{sec:time_avg}). Then, we quickly relate global jumps to local jumps (\autoref{sec:global_to_local_jumps}).
Then, in~\autoref{sec.oftprop}, we derive the refined properties of the operator Fourier transform and, most importantly, how it interplays with Gibbs conjugation. Lastly, 
We display the explicit Dirichlet form and its relation to the commutator (\autoref{sec:Dirichlet}). Based on these key ingredients, we present the proofs of the main results in~\autoref{sec:proofsthms}. In the appendix, we instantiate standard Lieb-Robinson bounds for the Lindbladians (\autoref{sec:LRestimates}), and show that adding a local gap condition allows us to bootstrap the quasi-polynomial runtime to logarithmic (\autoref{sec:improvegap}).

\acknowledgments 
We thank Itai Arad, Ainesh Bakshi, Shankar Balasubramanian, Thiago Bergamaschi, Fernando Brandao, András Gilyén, Aram Harrow, Ting-Chun Lin, Allen Liu, Ankur Moitra, Jiaqing Jiang, Isaac Kim, Michael Kastoryano, Sidhanth Mohanty, John Preskill, Amit Rajaraman, Daniel Ranard, Ewin Tang, Umesh Vazirani, David Wu, and Yongtao Zhan for helpful discussions. We thank Anurag Anshu and Quynh T. Nguyen for collaboration in related work. We thank Kohtaro Kato and Tomotaka Kuwahara for letting us know about their concurrent and independent work. CFC is supported by a Simons-CIQC postdoctoral fellowship through NSF QLCI Grant No. 2016245. CR acknowledges financial support from the ANR project QTraj (ANR-20-CE40-0024-01) of the French National Research Agency (ANR).
\section*{Data Availability}
No datasets were generated or analyzed during the current study. 
\section*{Conflict of interest statement}
The authors declare that they have no conflict of interest. 
\section*{Notations}\label{sec:recap_notation}
We use $a \lesssim b$ to absorb absolute constants.
We write scalars, functions, and vectors in normal font, matrices in bold font $\vO$, and superoperators in curly font~$\CL$ with matrix arguments in square brackets $\CL[\vrho]$. We use $\CO(\cdot),\Omega (\cdot)$ to denote asymptotic upper and lower bounds.
\begin{align}
\vI&: &\text{the identity operator}\\
\beta&: &\text{ inverse temperature}\\
\vrho_{\beta}&:= \frac{\e^{-\beta \vH }}{\tr[ \e^{-\beta \vH }]} (\equiv \vrho) \quad &\text{the Gibbs state with inverse temperature $\beta$}\\
A \subset \Lambda&:& \text{A subset of vertices}\\
\labs{A}&: & \text{cardinality of the region $A$}\\
n &= \labs{\Lambda} &\text{ system size (number of qubits) of the Hamiltonian $\vH$}\\
\{\vA^a\}_{a}&: &\text{set of jumps (for defining the Lindbladian)}\\
P_A&:& \text{set of nontrivial Pauli string on region $A$}\\
P^1_A&:=\{\vX_i,\vY_i,\vZ_i\}_{i\in A}  & \text{set of 1-local Pauli on region $A$}
\end{align}
Fourier transform notations:
\begin{align}
\vH &= \sum_i E_i \ketbra{\psi_i}{\psi_i}&\text{the Hamiltonian of interest and its eigendecomposition}\\
\vP_{E}&:= \sum_{i:E_i = E} \ketbra{\psi_i}{\psi_i}&\text{eigenspace projector for energy $E$}\\
\nu& \in B(\vH)  &\text{the set of Bohr frequencies, i.e., energy differences}\\
\vA_\nu&:=\sum_{E_2 - E_1 = \nu } \vP_{E_2} \vA \vP_{E_1} &\text{amplitude of $\vA$ that changes the energy by exactly $\nu$}\\
{\vA}(t) &:=\e^{i\vH t}\vA \e^{-i\vH t}& \text{Heisenberg-evolved operator $\vA$}\\
\hat{\vA}_{(f)}(\omega) &:= \frac{1}{\sqrt{2\pi}}\int_{-\infty}^{\infty} \e^{-\ri \omega t}f(t) \vA(t)\mathrm{d}t& \text{operator Fourier Transform for $\vA$ weighted by $f$}\\
\hat{f}(\omega)&=\lim_{K\rightarrow  \infty}\frac{1}{\sqrt{2\pi}}\int_{-K}^{K}\e^{-\ri\omega t} f(t)\mathrm{d}t & \text{the Fourier transform of function $f$}		
\end{align}
Norms: 
\begin{align}
	\norm{\vO}&:= \sup_{\ket{\psi},\ket{\phi}} \frac{\bra{\phi} \vO \ket{\psi}}{\norm{\ket{\psi}}\cdot \norm{\ket{\phi}}}= \norm{\vO}_{\infty} \quad &\text{the operator norm of a matrix $\vO$}\\
  	\norm{\vO}_p&:= (\tr \labs{\vO}^p)^{1/p} \quad&\text{the Schatten p-norm of a matrix $\vO$}\\
 \norm{\CL}_{p-p} &:= \sup_{\vO\ne 0} \frac{\normp{\CL[\vO]}{p}}{\normp{\vO}{p}}\quad&\text{the induced $p-p$ norm of a superoperator $\CL$}
\end{align}
\section{Properties of time-averaging}\label{sec:time_avg}
The key property we exploit from the time-averaging map is that it is simultaneously (quasi)-localized and stationary. Interestingly, the interplay between time-averaged dynamics, stationarity, and Dirichlet form has also been recently exploited in full glory in the recent analysis of classical slow-mixing Markov chains~\cite{liu2024locally}. 

Given any Lindbladian, we can define a time-averaged map $\CR_t$:
\begin{align}
    \CR_{t}[\cdot]:= \frac{1}{t}\int_{0}^t \exp\L(s\CL\R)[\cdot] \,\rd s.
\end{align}

\begin{lem}[Time-averaging implies approximate stationarity]\label{lem:time_average}
 For any Lindbladian $\CL$ and any operator $\vO$ such that $\norm{\vO}\le1$, the time-averaged map $\CR_t$ satisfies
 \begin{align}
     \lnorm{\CL^{\dagger}[\CR_t^{\dagger}[\vO]] } \, \le \frac{2}{t}.
 \end{align}
\end{lem}
\begin{proof}
Evaluate the integral 
    \begin{align}
        \frac{1}{t}\CL^{\dagger}\int_0^t \e^{\CL^{\dagger}s}\rd s = \frac{1}{t}(\e^{\CL^{\dagger}t}-\CI),\label{eq:L_integral}
    \end{align}
    and take suitable norms to conclude the proof.
\end{proof}
For our purposes, the approximate stationarity also holds for the Dirichlet form (\autoref{sec:Dirichlet}). 
\begin{cor}[Approximate stationarity of the Dirichlet form]\label{cor:dirichlet}
For any $\vO$ such that $\norm{\vO}\le 1$, consider a $\vrho$-detailed balance Lindbladian $\CL$. Then, the Dirichlet form $\CE(\vX,\vY):=-\langle \vX,\CL^\dagger(\vY)\rangle_{\vrho}$ for the time-averaged operator $\CR_t^{\dagger}[\vO]$ satisfies
\begin{align}
\CE(\CR^{\dagger}_{t}[\vO]) \le \frac{2}{t}.
\end{align}
\end{cor}
\begin{proof} 
Rewrite 
    \begin{align}
        \CE(\CR^{\dagger}_{t}[\vO]) =- \braket{\CR^{\dagger}_{t}[\vO],\CL^{\dagger}\CR^{\dagger}_{t}[\vO]}_{\vrho}
    \end{align}
and use~\autoref{lem:time_average}, \autoref{lem:operatornorm}, and $\norm{\CR^{\dagger}_{t}[\vO]}\le \norm{\vO} \le 1$ to conclude the proof.
\end{proof}

For our purposes, the above will be applied to the detailed-balanced Lindbladian~\eqref{eq:exact_DB_L} associated with jump operators being single-site Pauli operators $\vA^a \in P^1_A$ on region $A$~\eqref{eq:RAt}.

\subsection{Quasi-locality}
Since the jumps are restricted to region $A$, we expect the associated Lindbladian and the time-averaging map $\CR_{A,t}$ to also be quasi-local. We will work in the Heisenberg picture, with the superoperator norm $\norm{\cdot}_{\infty-\infty}$ induced by the operator norm.

\begin{lem}[Truncation error]\label{lem:truncation}
Consider the time averaging map $\CR^{\dagger}_{A,t,\ell}$ associated with the Hamiltonian $\vH_{\ell}$ containing all jumps $\vA^a$ with $a\in P^1_A$. Then,  
    \begin{align}
        \norm{\CR^{\dagger}_{A,t,\ell} - \CR^{\dagger}_{A,t}}_{\infty-\infty} \lesssim t \norm{\CL^{\dagger}_{A,\ell}-\CL_A^{\dagger}}_{\infty-\infty},
    \end{align}
    where $\mathcal{L}_A$, resp. $\mathcal{L}_{A,\ell}$, is the Lindbladian of the Gibbs sampler with Hamiltonian $\vH$ and jumps $\vA^a\in P_A^1$, resp. the one associated to the Hamiltonian $\vH_\ell$ with jumps $\vA^a\in P_A^1$.
\end{lem}
\begin{proof} We expand
    \begin{align}
        \CR^{\dagger}_{A,t,\ell} - \CR^{\dagger}_{A,t} &= \frac{1}{t}\int_0^t ( \e^{\CL^{\dagger}_{A,\ell}s}-\e^{\CL_{A}^{\dagger}s})\rd s\\
        &=\frac{1}{t}\int_0^t \int_{0}^{s} \e^{\CL_A^{\dagger}(s-s')}(\CL^{\dagger}_{A,\ell}-\CL_A^{\dagger})\e^{\CL^{\dagger}_{A,\ell}s'} \rd s'\rd s,
    \end{align}
    take the norms, and evaluate the time integrals to conclude the proof.
\end{proof}

The quasi-locality of the Lindbladian $\CL$ is a standard Lieb-Robinson argument (see~\autoref{sec:LRestimates}).
\begin{lem}[Quasi-locality]\label{lem:quasilocal} For a Hamiltonian $\vH$ with interaction degree at most $d$, denote by $\mathcal{L}^\dagger_{A,\ell}$ the generator of the Gibbs sampler with Metropolis weight~\eqref{eq:Metropolis}, jump operators ${\vA^a}\in P_A^1$ and Hamiltonian $\vH_{\ell}$ containing all terms $\vH_{\gamma}$ with distance $\operatorname{dist}(\gamma,A) < \ell-1$ from $A$. Then, for every $\ell \ge 4\e^2\beta d,$
\begin{align}
    \norm{\CL^{\dagger}_{A,\ell}-\CL_A^{\dagger}}_{\infty-\infty} \lesssim |A|\Big(\e^{-c'\frac{\ell}{d\beta}} + 2^{-\ell}\Big)\,, 
\end{align}    
for some universal constant $c'>0$.
\end{lem}

\begin{rmk}
    The tail falls exponentially fast with the distance $\ell$, while the truncation error in the time-averaged maps accumulates linearly with $t$. Thus, the quasi-locality holds for exponential times.
\end{rmk}

\section{From global to local jumps}\label{sec:global_to_local_jumps}

When reasoning about recovery maps in the proof of~\autoref{thm:main}, high-weight Paulis acting on region $A$ naturally appear in our arguments. However, it is desirable to have local Lindbladian jumps, as high-weight Paulis always incur an exponential time-overhead in the simulation. Here, we develop a method for reducing commutators of high-weight Pauli strings to those of local jumps of weight $1$.

For any region $A\subset \Lambda$, recall the set of single Paulis
\begin{align}
P_A^1 :=\{\vX_i,\vY_i,\vZ_i\}_{i\in A}    
\end{align}
with cardinality $\labs{P_A^1} = 3 \labs{A}.$ We begin with some additional reduction between global and local Paulis.

\begin{lem} For any set of operators $\vA^i\in P_A^1$ and $\vO$,  $[\prod^{w}_{i=1} \vA^{i}, \vO] = \prod^{j-1}_{i=1}\vA^{i} \cdot [\vA^{i}, \vO] \cdot \prod^{w}_{i=j} \vA^{i}$.
\end{lem}
\begin{cor}[Global to local commutators]\label{cor:global_local_comm}
For any Pauli string $\vS = \prod^{w}_{i=1} \vA^{i}$ with weight $w$, and any inverse temperature $\beta$ such that $\beta>{1/d}>0$
\begin{align}
    \norm{[\vS,\vO]}_{\vrho_\beta} \lesssim 2^{w} r'\sum_{j=1}^w \lnorm{[\vA^j,\vO]}_{\vrho_\beta}^{16\beta^2_0/\beta^2}
\end{align}    
for some constant $r'$ that depends on $\beta,\beta_0,\sigma$. Whenever $\beta\le {1/d}$, we have instead
\begin{align}
   \norm{[\vS,\vO]}_{\vrho_\beta} \lesssim 2^w\,\sum_{j=1}^w \lnorm{[\vA^j,\vO]}_{\vrho_\beta}\,.
\end{align}
\end{cor}
\begin{proof} Invoke~\autoref{lem:weight_norm_comm} and~\autoref{cor:highweight} below twice, with $\beta_0={1/4d}$ and denoting by $r$ a constant that depends on $\beta,\beta_0,\sigma$ which may change from line to line,
    \begin{align}
    \norm{[\vS,\vO]}_{\vrho_\beta} &\le r\,\sum_{j=1}^w{2^{j-1}} \lnorm{[\vA^j,\vO]\prod_{i=j}^w\vA^i}_{\vrho_\beta}^{\frac{4\beta_0}{\beta}}\\
    &\le r\,\sum_{j=1}^w {2^{j-1}}\L( \L(2^{w-j+1}\lnorm{[\vA^j,\vO]}_{\vrho_\beta}\R)^{\frac{4\beta_0}{\beta}} \R)^{\frac{4\beta_0}{\beta}}\le 2^{w} r'\sum_{j=1}^w \lnorm{[\vA^j,\vO]}_{\vrho_\beta}^{16\beta^2_0/\beta^2}
\end{align}
to conclude the proof of the first bound. The second bound follows even more easily after replacing the use of \autoref{lem:weight_norm_comm} by that of \autoref{lem:loose_AO}. 
\end{proof}
\begin{rmk}
    Without this reduction step to local jumps, the proof of~\autoref{thm:main} still goes through with similar parameters. However, the exponential time-overhead is fundamental with high-weight Pauli and cannot be improved by any additional mixing time assumption. Even though this reduction step does incur a $2^{w}$ multiplicative factor, the effect of this factor may be polynomial assuming a suitable local gap, see \autoref{sec:improvegap}.
\end{rmk}

\section{Regularizing the operator FT at low-temperatures}\label{sec.oftprop}
At low enough constant temperatures, the complex time dynamics can be very wild $\norm{\e^{\beta \vH}\vA \e^{-\beta \vH}}\ge \e^{c n}$ in more than one spatial dimension. 
It will be tremendously helpful to decompose the operator over operator Fourier transforms at different Bohr frequencies.
\begin{lem}[Decomposing an operator by the energy change]\label{lem:sumoverenergies}
For any (not necessarily Hermitian) operator $\vA$, we have that
\begin{align}
\vA =  \frac{1}{\sqrt{2\sigma\sqrt{2\pi}}} \int_{-\infty}^{\infty} \hat{\vA}(\omega)\rd \omega.
\end{align}
\end{lem}
\begin{proof}
    \begin{align}\vspace{-1cm}
        \int_{-\infty}^{\infty} \hat{\vA}(\omega)\rd \omega = \int_{-\infty}^{\infty} \sum_{\nu} \vA_{\nu} \hat{f}(\omega-\nu) \rd \omega= \sum_{\nu} \vA_{\nu} \int_{-\infty}^{\infty}\hat{f}(\omega-\nu) \rd (\omega-\nu)
        = \sqrt{2\pi} f(0) = \sqrt{2\sigma\sqrt{2\pi}}.
    \end{align}
\end{proof}

The Gaussian damping has a regularization effect due to its super-exponential decay. 

\begin{lem}[Norm bounds on imaginary time conjugation]\label{lem:bounds_imaginary_conjugation}
For any $\beta ,\omega\in \BR$ and operator $\vA$ with norm $\norm{\vA}\le 1$, the operator Fourier transform $\hat{\vA}(\omega)$ with uncertainty $\sigma$~\eqref{eq:OFT},~\eqref{eq:fwft} satisfies 
\begin{align}
        \e^{\beta \vH} \hat{\vA}(\omega) \e^{-\beta \vH}
& = \e^{\beta\omega} \cdot \hat{\vA}(\omega+2\sigma^2\beta) \e^{\sigma^2\beta^2}.
\end{align}
Thus, 
\begin{align}
    \norm{\e^{\beta \vH} \hat{\vA}(\omega) \e^{-\beta \vH}} \le \frac{\e^{\sigma^2\beta^2}}{\sqrt{{\sigma}\sqrt{2\pi}}} \e^{\beta\omega}.
\end{align}
\end{lem}
In comparison, directly conjugating the unfiltered operator could yield a norm $\norm{\e^{\beta \vH} \hat{\vA}(\omega) \e^{-\beta \vH}}$ growing with the system size $n$; the Gaussian filtering centered at Bohr frequency $\omega$ removes the dependence on the system size $n$, and only depends the Bohr frequency $\omega.$ While it still grows exponentially, the bounds are now entirely (quasi)-local.

\begin{proof}
Recall
\begin{align}
        \e^{\beta \vH} \hat{\vA}(\omega) \e^{-\beta \vH} &= \sum_{\nu} \vA_{\nu}\frac{1}{\sqrt{{\sigma}\sqrt{2\pi}}} \exp\L(- \frac{(\omega-\nu)^2}{4\sigma^2}\R)\e^{\beta \nu}\\
&=  \sum_{\nu} \vA_{\nu}\frac{1}{\sqrt{{\sigma}\sqrt{2\pi}}} \exp\L(- \frac{(\omega+2\sigma^2\beta-\nu)^2}{4\sigma^2}+\beta\omega+\sigma^2\beta^2\R) \\
& = \hat{\vA}(\omega+2\sigma^2\beta) \cdot \e^{\beta\omega+\sigma^2\beta^2}.
\end{align}
Apply triangle inequality to the integral $\norm{\hat{\vA}(\omega)} \le\frac{1}{\sqrt{2\pi}} \int_{-\infty}^{\infty}\labs{f(t)}\rd t  = \frac{1}{\sqrt{{\sigma}\sqrt{2\pi}}}$ to conclude the proof.
\end{proof}
At high enough temperatures, there is a stronger bound (within the convergence radius of the Taylor expansion) that exploits the bounded interaction degree of the Hamiltonian.
\begin{lem}[Convergence for imaginary time]\label{lem:convergence_imaginary}
    For Hamiltonians defined in~\autoref{sec:Ham} with interaction degree at most $d$, a single-site operator $\norm{\vA}\le1$, and $\labs{\beta}< 1/2d,$
    \begin{align}
        \norm{\e^{\beta \vH}\vA\e^{-\beta \vH}}\le \frac{1}{1-2d\labs{\beta}}.
    \end{align}
\end{lem}
\begin{proof}
    Talyor-expand into nested commutators
    \begin{align}
\e^{\beta \vH}\vA\e^{-\beta \vH} &= \vA + \beta [\vH,\vA] + \frac{\beta^2}{2!}[\vH,[\vH,\vA]]+\cdots\\
&=: \sum_{k=0}^\infty \frac{\beta^k}{k!}\CC_{\vH}^{k}[\vA].
    \end{align}

Since the operator $\vA$ is single site, for any string $[\vH_{\gamma_k},\cdots, [\vH_{\gamma_{1}}, \vA ]]$ in $\CC^{k-1}_{\vH}[\vA]$, the outermost commutator $[\vH,\cdot ]$ has at most $\max( (k-1)d,d)\le kd$ Hamiltonian terms that may contribute, we have that
\begin{align}
\norm{\CC_{\vH}^{k}[\vA]}
\le k! (2d)^k
\end{align}
sum over the geometric series to conclude the proof.
\end{proof}

From the above, we can also extend to higher weight Paulis by expanding them into products of single-site Paulis. The bound grows exponentially with the weight but is independent of the global system size.
\begin{cor}[High weight]\label{cor:highweight}
In the setting of~\autoref{lem:convergence_imaginary}, for any Pauli string $\vS$ of weight $w$,
\begin{align}
    \norm{\e^{\beta \vH}\vS\e^{-\beta \vH}} \le \L(\frac{1}{1-2d\labs{\beta}}\R)^{w}.
\end{align}
\end{cor}

Using the above, we bootstrap for an even better norm bound for the Operator Fourier Transform.
\begin{cor}[Norm decay for large energy difference]\label{cor:norm_decay}
For any $\beta ,\omega\in \BR$ and operator $\vA$, we have that
\begin{align}
    \norm{\hat{\vA}(\omega)} \le \frac{\e^{-\beta \omega + \sigma^2 \beta^2}}{\sqrt{{\sigma}\sqrt{2\pi}}} \norm{\e^{\beta \vH}\vA\e^{-\beta \vH}}.
\end{align}
\end{cor}
\begin{proof} ``Borrow'' cancelling factors of $\e^{\beta \vH}$ on the left and right
    \begin{align}
\hat{\vA}(\omega)  &= \e^{-\beta \vH}\cdot (\e^{\beta \vH}\hat{\vA}(\omega)\e^{-\beta \vH}) \cdot\e^{\beta \vH}\\
&=\e^{-\beta \vH}\cdot( \widehat{[\e^{\beta \vH}\vA\e^{-\beta \vH}]}(\omega) )\cdot \e^{\beta \vH}\tag*{(For any $\vH$, operator FT commutes with imaginary time conjugation)}
\end{align}    
and apply~\autoref{lem:bounds_imaginary_conjugation} for $\vA'= \e^{\beta \vH}\vA\e^{-\beta \vH}$ to conclude the proof.
\end{proof}
This will allow us to truncate the Bohr frequencies with an exponentially small error.

\subsection{Controlling commutators within weighted norms}

The main goal here is to control the effect of taking commutators\textit{inside} the weighted norm. A direct H\:{o}lder's inequality can give a loose bound.
\begin{lem}[Loose bounds for high temperature]\label{lem:loose_AO}
For any operator $\vA,\vO$, and full rank $\vrho,$
\begin{align}
\lnormp{[\vA,\vO]}{\vrho} \le \L(\norm{\vrho^{1/4}\vA\vrho^{-1/4}}+\norm{\vrho^{-1/4}\vA\vrho^{1/4}}\R) \lnormp{\vO}{\vrho}.
\end{align}     
\end{lem}
\begin{proof} Expand the commutator
    \begin{align}
        \lnormp{[\vA,\vO]}{\vrho} \le \lnormp{\vA\vO}{\vrho} + \lnormp{\vO\vA}{\vrho}
    \end{align}
    and rewrite the weighted norm in the Frobenius norm 
    \begin{align}
        \norm{\vA\vO}_{\vrho} = \norm{ \vrho^{1/4}\vA\vO\vrho^{1/4}}_2 =\norm{ \vrho^{1/4}\vA\vrho^{-1/4}\cdot\vrho^{1/4}\vO\vrho^{1/4}}_2&\le  \norm{ \vrho^{1/4}\vA\vrho^{-1/4}}\cdot\norm{\vrho^{1/4}\vO\vrho^{1/4}}_2 \\
        &= \norm{ \vrho^{1/4}\vA\vrho^{-1/4}}\norm{\vO}_{\vrho}.
    \end{align}
Repeat for $\norm{\vO\vA}_{\vrho}$ to conclude the proof. 
\end{proof}

\begin{rmk}
    If $\vA$ is Hermitian, then $\norm{\vrho^{1/4}\vA\vrho^{-1/4}}$ = $\norm{\vrho^{-1/4}\vA\vrho^{1/4}}$. Otherwise, they may differ.
\end{rmk}

Of course, the above naive bound may diverge at low temperatures due to the imaginary time conjugation $\norm{\vrho^{1/4}\vA\vrho^{-1/4}}$, which may grow poorly with the system size (especially when the Hamiltonian is noncommuting).
To give a convergent bound at low temperatures, the operator Fourier transforms will become very handy for regularizing divergences.

\begin{lem}[Bounds on multiplication]\label{lem:weight_norm_comm}
For any operators normalized by $\norm{\vO}\le 1$, $\norm{\vA}\le 1$, and any pair of inverse temperatures $\beta_0, \beta$ such that $\beta >4 \beta_0>0$,
\begin{align}
\norm{\vA\vO}_{\vrho_\beta},\norm{\vO\vA}_{\vrho_\beta}\lesssim 
    \Big(\frac{\e^{\sigma^2\beta^{'2}}}{\beta'{\sigma}}+\frac{\e^{\sigma^2\beta^{2}_0}}{\beta_0{\sigma}}\Big)\norm{\vO}_{\vrho_\beta}^{\frac{4\beta_0}{\beta}}\L(\norm{\vrho_{\beta_0}\vA\vrho_{\beta_0}^{-1}}+\norm{\vrho_{\beta_0}^{-1}\vA\vrho_{\beta_0}}\R) 
\end{align}
where $\beta' := \beta /4 -\beta_0.$ 
\end{lem}
\begin{rmk}
    As $\beta \rightarrow 4 \beta_0,$ we see that the exponent approaches unity $\frac{4\beta_0}{\beta}\rightarrow 1$ and almost recovers the loose bound (\autoref{lem:loose_AO}). 
\end{rmk}
For the RHS to be useful, one should take the largest possible $\beta_0$ according to the available bounds on the imaginary time conjugation~\autoref{lem:convergence_imaginary}.
\begin{proof}

Introducing a decomposition of operator $\vA$ by the Bohr frequencies and a tunable truncation parameter $\Omega>0$, we get
    \begin{align}
        c\lnormp{\vA\vO}{\vrho_\beta} &= \lnormp{\int_{-\infty}^{\infty}\hat{\vA}(\omega)\rd \omega\vO}{\vrho_\beta} \tag*{(\autoref{lem:sumoverenergies}, $c = \sqrt{2\sigma\sqrt{2\pi}}$)}\\
        &\le \lnormp{\int_{\labs{\omega}\le \Omega}\hat{\vA}(\omega)\rd \omega\vO}{\vrho_\beta } + \lnormp{\int_{\labs{\omega}> \Omega}\hat{\vA}(\omega)\rd \omega\vO}{\vrho_\beta}.
\end{align}

We bound the two terms using different bounds tailored to different regimes. 
For $\labs{\omega}\le \Omega$, we would like to express in terms of the weighted norm $\norm{\vO}_{\vrho_\beta}$
\begin{align}
        \lnormp{\int_{\labs{\omega}\le \Omega}\hat{\vA}(\omega)\rd \omega\vO}{\vrho_\beta} 
        &\le \lnormp{\vO}{\vrho_\beta}\int_{\labs{\omega} \le \Omega } \norm{\vrho_\beta^{1/4}\hat{\vA}(\omega)\vrho_\beta^{-1/4}} \rd \omega \tag*{(\autoref{lem:loose_AO})} \\
        & \lesssim \lnormp{\vO}{\vrho_\beta} \int_{-\Omega}^{\Omega} \frac{1}{\sqrt{\sigma}} \e^{-\beta'\omega+\sigma^2\beta^{'2}} \norm{\vrho_{\beta_0}\hat{\vA}(\omega)\vrho_{\beta_0}^{-1}}\rd \omega  \tag*{(\autoref{cor:norm_decay}, $\beta' := \beta/4 -\beta_0$)}\,.
\end{align}
 We see that we paid a price of $\e^{\beta' \Omega}$ for inverting the Gibbs state, so we cannot choose arbitrarily large $\Omega$.
For $\labs{\omega} \ge \Omega$, we cannot afford to invert the Gibbs state anymore, so we drop the weighted norm 
\begin{align}
    \lnormp{\int_{\labs{\omega}> \Omega}\hat{\vA}(\omega)\rd \omega\vO}{\vrho_\beta}
    &\lesssim \int_{\labs{\omega}> \Omega} \norm{\hat{\vA}(\omega)} \rd \omega\tag*{(Using \autoref{lem:operatornorm} and $\norm{\vO} \le 1$)}\\
    &\lesssim \int_{\labs{\omega}> \Omega} \frac{\e^{-\beta_0 \labs{\omega} + \sigma^2 \beta_0^2}}{{\sqrt{\sigma}}} \L(\norm{\e^{\beta_0 \vH}\vA\e^{-\beta_0 \vH}}+\norm{\e^{-\beta_0 \vH}\vA\e^{\beta_0 \vH}}\R)  \rd \omega \tag*{(\autoref{cor:norm_decay})}\\
    &\lesssim \e^{-\beta_0 \Omega} \frac{\e^{\sigma^2 \beta_0^2}}{\beta_0\sqrt{\sigma}}\L(\norm{\vrho_{\beta_0}^{-1}\vA\vrho_{\beta_0}}+\norm{\vrho_{\beta_0}\vA\vrho^{-1}_{\beta_0}}\R) \label{eq:greater_Omega}.
\end{align}
Balance both terms by setting 
\begin{align}
    \e^{\Omega} = \L(\frac{1}{\norm{\vO}_{\vrho_\beta}}\R)^{\frac{4}{\beta}},
\end{align}
and rearrange the factor $1/c$ to conclude the proof.
\end{proof}

\section{Dirichlet forms}\label{sec:Dirichlet}
A central object in the analysis of classical Markov chains is the Dirichlet form. In this section, we write down useful equivalent versions of the Dirichlet forms for the exactly detailed balanced Lindbladian; our argument would not be possible without the explicit forms provided in~\cite{rouze2024efficient}. 

\begin{lem}[{\cite[Lemma C.2]{rouze2024efficient}}]\label{lem:alphabar_Dirichlet}
Suppose the transition $\CT$ part of a $\vrho_{\beta}$-detailed balance Lindbladian at inverse temperature $\beta$ can be written as bilinear combination of $\vA_{\nu_1}, \vA_{\nu_2}^{\dagger}$
\begin{align}
    \CT [\cdot]:= \sum_{\nu_1,\nu_2} \alpha_{\nu_1,\nu_2} \vA^a_{\nu_1}[\cdot]\vA^{a\dagger}_{\nu_2}, \quad \text{and}\quad h_{\nu_1,\nu_2} = \alpha_{\nu_1,\nu_2}\e^{\beta(\nu_1+\nu_2)/4}.
\end{align}
Then, the Dirichlet form $\CE(\vX,\vY)$ for any two operator $\vX,\vY$ can be written as
\begin{align}
\CE(\vX,\vY):=-\braket{\vX,\CL^{\dagger}[\vY]}_{\vrho_\beta} = \sum_{a}\sum_{\nu_1,\nu_2\in B} \bar{\alpha}_{\nu_1,\nu_2} \tr\L[\sqrt{\vrho_{\beta}}[\vA^a_{\nu_1},\vX ]^{\dagger}\sqrt{\vrho_{\beta}}[\vA^a_{\nu_2},\vY ] \R] =:\sum_{a} \CE_a(\vX,\vY)
\end{align}
with
\begin{align}
    \bar{\alpha}_{\nu_1,\nu_2} := h_{\nu_1,\nu_2} \frac{1}{2\cosh(\beta(\nu_1-\nu_2)/4)} = \bar{\alpha}_{-\nu_1,-\nu_2}.
\end{align}
\end{lem}
In this paper, we sometimes denote the Dirichlet form evaluated on the same two operators $(\vX,\vY)=(\vX,\vX)$ as 
\begin{align}
    \CE(\vX)\equiv\CE(\vX,\vX).
\end{align}
In the above, the statement holds for any detailed-balanced Gibbs sampler whose transition part takes the form of $\sum_{\nu_1,\nu_2} \alpha_{\nu_1,\nu_2} \vA^a_{\nu_1}[\cdot]\vA^{a\dagger}_{\nu_2}$~\cite{ding2024efficient,gilyen2024quantum}, but for concreteness, we recall the explicit forms for $h_{\nu_1,\nu_2}$
\begin{align}
h_{\nu_1,\nu_2}=\int_{\frac{\beta\sigma^2}{2}}^\infty\, \frac{1}{{2}}\sqrt{\frac{\beta}{\pi x}}\e^{-\frac{\beta(\nu_1+\nu_2)^2}{16 x}-\frac{\beta x}{4}}dx\,\e^{-\frac{(\nu_1-\nu_2)^2}{8\sigma^2}} \tag*{(the Metropolis weight~\eqref{eq:Metropolis})},
\end{align}
\begin{align*}
h^G_{\nu_1,\nu_2}=\frac{\sigma_\gamma}{{\sqrt{\sigma^2+\sigma_\gamma^2}}}\e^{-\frac{(\nu_1+\nu_2)^2+(2\omega_\gamma)^2}{8(\sigma^2+\sigma_\gamma^2)}}\e^{-\frac{(\nu_1-\nu_2)^2}{8\sigma^2}}
\tag*{(the Gaussian weight \eqref{eq:Gaussianweight})}.
\end{align*}
However, for our usage, instead of eigenoperators $\vA_{\nu}$'s, we need to further rewrite the Dirichlet form in terms of the more physical operator Fourier transform $\hat{\vA}(\omega)$'s. We begin with the auxiliary calculation for the Gaussian weight.

\begin{lem}[Dirichlet form with explicit operator Fourier transforms]\label{lem:Dirichlet_Gaussian}
    The Dirichlet form for the Lindbladian~\eqref{eq:exact_DB_L} with Gaussian weight $\gamma^{G}(\omega)$~\eqref{eq:Gaussianweight} can be rewritten as
\begin{align}
    \CE(\vX,\vY) &= \sum_a \int_{-\infty}^{\infty}\int_{-\infty}^{\infty} g(t) h^G(\omega)\tr\L[\sqrt{\vrho_{\beta}}[\hat{\vA}^a(\omega,t),\vX ]^{\dagger}\sqrt{\vrho_{\beta}}[\hat{\vA}^a(\omega,t),\vY ] \R]\rd t \rd \omega,
\end{align}
where $\hat{\vA}(\omega,t):=\e^{i\vH t}\hat{\vA}(\omega)\e^{-i\vH t}$, $h^G(\omega)=\e^{-{\frac{\beta\omega_\gamma}{4}}}\e^{-{\omega^2/2\sigma_\gamma^2}} \ge 0$ and $g(t) = \frac{1}{\beta\cosh(2\pi t/ \beta)}\ge 0.$
\end{lem}
When evaluated on the same operator $\CE(\vX,\vX)$, the nice feature of the above is that it is an integral of \textit{non-negative} summand; if the Dirichlet form is small, we must have that the integrand is also small.
\begin{proof}

By linearity, it suffices to prove for single jump $\vA$. Consider the time-domain expression for 
\begin{align}
    \frac{1}{2\cosh(\beta(\nu_1-\nu_2)/4)} = \int_{-\infty}^{\infty}g(t)\e^{-\ri (\nu_1-\nu_2) t} \rd t\quad \text{where} \quad g(t) = \frac{1}{\beta\cosh(2\pi t/ \beta)}\ge 0.
\end{align}
Then, we may rewrite the expression in~\autoref{lem:alphabar_Dirichlet} as
\begin{align}
    \CE(\vX,\vY) &= \int_{-\infty}^{\infty}\sum_{\nu_1,\nu_2\in B}g(t) h^G_{\nu_1,\nu_2} \tr\L[\sqrt{\vrho_{\beta}}[\vA_{\nu_1}\e^{\ri \nu_1t},\vX ]^{\dagger}\sqrt{\vrho_{\beta}}[\vA_{\nu_2}\e^{\ri \nu_2t},\vY ] \R]\rd t.
\end{align}
We rewrite the bilinear sums in terms of operator Fourier transforms: 

\begin{align}
\sum_{\nu_1,\nu_2\in B} h^G_{\nu_1,\nu_2}\vA_{\nu_1}(\cdot) \vA^\dagger_{\nu_2}
&=\sum_{\nu_1,\nu_2\in B}\e^{\frac{\beta(\nu_1+\nu_2)}{4}} \alpha^G_{\nu_1,\nu_2}\vA_{\nu_1}(\cdot) \vA^\dagger_{\nu_2}\\
&=\int_{-\infty}^\infty\, \gamma(\omega)\vrho_\beta^{-\frac{1}{4}}\hat{\vA}(\omega)\vrho_\beta^{\frac{1}{4}}(\cdot)\vrho_\beta^{\frac{1}{4}}\hat{\vA}(\omega)^\dagger\vrho_\beta^{-\frac{1}{4}}\, \rd\omega\\
&=\int_{-\infty}^{\infty}  \e^{- \frac{(\omega+\omega_{\gamma})^2}{2\sigma_\gamma^2}}\e^{\beta\vH/4}\hat{\vA}(\omega)\e^{-\beta\vH/4}(\cdot) \e^{-\beta\vH/4}\hat{\vA}(\omega)^{\dagger}\e^{\beta\vH/4}\rd\omega
\end{align}
\begin{align}
    (cont.)&\qquad\qquad\qquad=\e^{\beta^2\sigma^2/8}\int_{-\infty}^{\infty}  \e^{- \frac{(\omega+\omega_{\gamma})^2}{2\sigma_\gamma^2}+\beta\omega/2}\hat{\vA}(\omega+\sigma^2\beta/2)(\cdot) \hat{\vA}(\omega+\sigma^2\beta/2)^{\dagger}\rd\omega\tag*{(By~\autoref{lem:bounds_imaginary_conjugation})}\\
&\qquad\qquad\qquad= \e^{\beta^2\sigma^2/8+\beta^2\sigma^4/8\sigma_\gamma^2-\omega_{\gamma}^2/2\sigma_\gamma^2}\int_{-\infty}^{\infty}  \e^{- \frac{(\omega+\beta\sigma^2/2)^2}{2\sigma_\gamma^2}}\hat{\vA}(\omega+\sigma^2\beta/2)(\cdot) \hat{\vA}(\omega+\sigma^2\beta/2)^{\dagger}\rd\omega\\
&\qquad\qquad\qquad= \e^{-{\frac{\beta\omega_{\gamma}}{4}}}\cdot  \int_{-\infty}^{\infty}  \e^{- \frac{\omega^2}{2\sigma_\gamma^2}}\hat{\vA}(\omega)(\cdot) \hat{\vA}(\omega)^{\dagger}\rd\omega \tag*{(Shift of integration range)}.
\end{align}

The last two lines are a simplification due to the identity
$\beta (\sigma^2+\sigma_\gamma^2) :=2 \omega_\gamma$. The result simply follows after observing that $\vA_\nu \e^{i\nu t}$ is the Fourier coefficient of $\vA(t):=\e^{i\vH t}\vA\e^{-i\vH t}$ corresponding to the Bohr frequency $\nu$.

\end{proof}

The Gaussian calculation also informs the Metropolis case by taking  a weighted linear combination~\cite[Proposition II.4]{chen2023efficient} 
\begin{align}
      \exp\L(-\beta\max\left(\omega +\frac{\beta \sigma^2}{2},0\right)\R) &= \gamma(\omega) = \int_{\frac{\beta\sigma^2} {2}}^{\infty}g_x \gamma^G_{x}(\omega) \rd x \\
      \quad \text{where} \quad (\omega_{\gamma}(x),\sigma_{\gamma}(x)) &= (x,\sqrt{2x/\beta - \sigma^2})\quad \text{and}\quad g_x:= \frac{1}{\sqrt{2\pi}}\frac{1}{\sigma_{\gamma}}\label{eq:omega_sigma_g}.
\end{align}
That is, the weight is inversely proportional to the width $\sigma_{\gamma}$ of the transition weight. Thus,
\begin{align}
     \sum_{\nu_1,\nu_2\in B} h_{\nu_1,\nu_2}\vA_{\nu_1}(\cdot) \vA^\dagger_{\nu_2} 
    &= \int_{\frac{\beta\sigma^2}{2}}^{\infty} g_x \e^{-\frac{\omega_{\gamma}^2}{2(\sigma_{\gamma}^2+\sigma^2)}}\L( \int_{-\infty}^{\infty}  \e^{- \frac{\omega^2}{2\sigma_\gamma^2}}\hat{\vA}(\omega)\cdot \hat{\vA}(\omega)^{\dagger}\rd\omega\R) \rd x\\
    &= \int_{-\infty}^{\infty} \undersetbrace{=:h(\omega)}{\L(\int_{\frac{\beta\sigma^2}{2}}^{\infty} g_x \e^{-\frac{\omega_{\gamma}^2}{2(\sigma_{\gamma}^2+\sigma^2)}}\e^{- \frac{\omega^2}{2\sigma_\gamma^2}}  \rd x \R)} \hat{\vA}(\omega)\cdot \hat{\vA}(\omega)^{\dagger}\rd\omega.
\end{align}
In the next Lemma, we compute the function $h(\omega)$ explicitly. Later, we will use the fact that the function decays exponentially with $\omega$.
\begin{lem}[Dirichlet form with the Metropolis weight] \label{lem:Dirichlet_Gaussianmetropolis} The Dirichlet form for the Metropolis weight is the same as Gaussian (\autoref{lem:Dirichlet_Gaussian}) up to replacing the function $h^G(\omega)$ by   
    \begin{align}
        h(\omega): = \e^{-\frac{\sigma^2\beta^2}{8}} \e^{-\labs{\omega}\beta/2}\ge 0.
    \end{align}
\end{lem}
\begin{proof}
    \begin{align}
    \int_{\frac{\beta\sigma^2}{2}}^{\infty} g_x \exp\L( -\frac{\omega_{\gamma}^2}{2(\sigma_{\gamma}^2+\sigma^2)} - \frac{\omega^2}{2\sigma_\gamma^2} \R)\rd x
    & = \int_{\frac{\beta\sigma^2}{2}}^{\infty} \frac{1}{\sqrt{2\pi} \sqrt{\frac{2x}{\beta}-\sigma^2}} \exp\L(-\frac{\beta x}{4} - \frac{\omega^2}{2(2x/\beta-\sigma^2)}\R) 
 \rd x\tag*{(Simplify)}\\
    & = \int_{0}^{\infty} \frac{1}{\sqrt{\pi} \sqrt{y}} \exp\L(-y - \frac{\sigma^2\beta^2}{8} -\frac{\omega^2\beta^2}{16y} \R) \rd y \tag*{(Let $y:=\beta x/4 - \sigma^2\beta^2/8$)}\\
    &=  \frac{2\e^{-\frac{\sigma^2\beta^2}{8}}}{\sqrt{\pi}} \int_{0}^{\infty} \exp\L(-s^2 - \frac{\omega^2\beta^2}{16s^2}\R) \rd s \tag*{(Let $s = \sqrt{y}$)}\\
    & = \e^{-\frac{\sigma^2\beta^2}{8}} \cdot \e^{-\labs{\omega}\beta/2}\tag*{(Since $\int_{0}^{\infty} \e^{-s^2-a^2/s^2} \rd s = \frac{\sqrt{\pi}}{2}\e^{-2\labs{a}}$)},
\end{align}
 as advertised.
\end{proof}

\subsection{Small Dirichlet form implies small commutators}
Here, we exploit the ``integral of squares'' structure of the Dirichlet form to control the commutator. Technically, the Dirichlet forms only tells us about $[\hat{\vA}^a(\omega,t),\vX ]$ for almost all $t,\omega,$ so we also need some continuity argument to control $[\hat{\vA}^a(\omega),\vX ]=[\hat{\vA}^a(\omega,0),\vX ]$ by sending $t\rightarrow 0$.

\begin{lem}[Bounding commutators by Dirichlet forms]\label{lem:commutator_Dirichlet}
For operators $\vA, \vO$ normalized by $\norm{\vA}, \norm{\vO}\le 1$, and any $\beta, \beta_0 >0$, 
\begin{align}
    \lnorm{[\vA,\vO]}_{\vrho_\beta} &\lesssim d^2\labs{A}\L(\frac{\e^{\sigma^2 \beta_0^2}}{\beta_0\sigma}+ \frac{\e^{\sigma^2\beta^2/16}}{\sqrt{g(1)\beta\sigma}} \R)^{\frac{\beta+4\beta_0}{\beta+5\beta_0}} \L(\norm{\vrho_{\beta_0}\vA\vrho_{\beta_0}^{-1}}+\norm{\vrho_{\beta_0}^{-1}\vA\vrho_{\beta_0}}\R)^{\frac{\beta}{\beta+5\beta_0}} \CE(\vO) ^{\frac{2\beta_0}{\beta+5\beta_0}}
\end{align}    
where $\labs{A}$ is the size of the support of $\vA$, $d$ is the interaction degree of the Hamiltonian $\vH$, and $\CE$ is the Dirichlet form associated with Linbladian~\eqref{eq:exact_DB_L} with Metropolis weight~\eqref{eq:Metropolis} and unique jump $\vA$ at the inverse temperature $\beta$.  
\end{lem}
\begin{proof}
 The strategy is to rewrite the commutator $[\vA,\vO]$ in terms of $[\vA(\omega,t),\vO]$ to relate to the Dirichlet form. To ease the notation, we write $\vrho\equiv \vrho_\beta$.

\textbf{STEP 1: extend to finite time interval.}
Since the time integral in the Dirichlet form at $t=0$ has measure zero, we control a closely related quantity $\vA(t):=\e^{i\vH t}\vA \e^{-i\vH t}$ for small tunable $\labs{t}\le \epsilon$
\begin{align}
    2\epsilon \lnorm{[\vA,\vO]}_{\vrho} =\lnorm{\int_{-\epsilon}^{\epsilon}[\vA,\vO]\rd t}_{\vrho} &\le \lnorm{\int_{-\epsilon}^{\epsilon}[\vA-\vA(t),\vO]\rd t}_{\vrho} + \lnorm{\int_{-\epsilon}^{\epsilon}[\vA(t),\vO]\rd t}_{\vrho}.
\end{align}
The first term can be bounded by
\begin{align}
    \lnormp{ \int_{-\epsilon}^{\epsilon}(\vA-\vA(t))\rd t}{\vrho} = \lnormp{\int_{-\epsilon}^{\epsilon}\L(\int_{0}^{t} \int_{0}^{t_1}[\vH,[\vH,\vA]](t_2)\R)\rd t_2\rd t_1\rd t}{\vrho} \le \frac{\epsilon^3}{3} \norm{[\vH,[\vH,\vA]]} \lesssim \epsilon^3 d^2 \labs{A}\label{eq:eps3}
\end{align}
where the first-order Taylor series vanishes due to the symmetry of the integral, and the last inequality uses the degree of the interaction graph to bound the commutator. 

\textbf{STEP 2: truncate the frequencies.}
Next, we move on to control the second term by splitting the integral (\autoref{lem:sumoverenergies}, $c = \sqrt{2\sigma\sqrt{2\pi}}$)
\begin{align}
         c\lnorm{\int_{-\epsilon}^{\epsilon}[\vA(t),\vO]\rd t}_{\vrho} &= \lnorm{\int_{-\epsilon}^{\epsilon} \int_{-\infty}^{\infty}[\hat{\vA}(\omega,t),\vO]\rd \omega\rd t}_{\vrho}\\
        &\le \lnorm{\int_{-\epsilon}^{\epsilon} \int_{\labs{\omega}\le \Omega}[\hat{\vA}(\omega,t),\vO]\rd \omega \rd t}_{\vrho} + \lnorm{\int_{-\epsilon}^{\epsilon}\int_{\labs{\omega}> \Omega}[\hat{\vA}(\omega,t),\vO] \rd \omega\rd t}_{\vrho}.
    \end{align}
We use the Cauchy-Schwarz inequality to bound the first term by the Dirichlet form
\begin{align}
    \lnorm{\int_{-\epsilon}^{\epsilon} \int_{\labs{\omega}\le \Omega}[\hat{\vA}(\omega,t),\vO]\rd \omega \rd t}_{\vrho} &\le \int_{-\epsilon}^{\epsilon} \int_{\labs{\omega}\le \Omega} \lnorm{[\hat{\vA}(\omega,t),\vO]}_{\vrho} \rd \omega \rd t\\
    &\le \int_{-\infty}^{\infty} \int_{-\infty}^{\infty} \frac{\indicator(\labs{t}\le \epsilon, \labs{\omega}\le \Omega)}{\sqrt{g(t)h(\omega)}} \cdot \sqrt{g(t)h(\omega)} \lnorm{[\hat{\vA}(\omega,t),\vO]}_{\vrho} \rd \omega \rd t\\
    &\le  \sqrt{\int_{-\infty}^{\infty} \int_{-\infty}^{\infty} \frac{\indicator(\labs{t}\le \epsilon, \labs{\omega}\le \Omega)}{g(t)h(\omega)} \rd t\rd \omega} \cdot \sqrt{\CE(\vO)},\tag*{(Cauchy-Schwarz)}\\
    &\le \frac{2\sqrt{2\epsilon}}{\sqrt{g(\epsilon)\beta}} \e^{\sigma^2\beta^2/16} \e^{\Omega\beta/4} \cdot \sqrt{\CE(\vO)}.\label{eq:sqrtE}
\end{align}
In the above, the diverging reciprocal $1/h(\omega) = \e^{\sigma^2\beta^2/8} \e^{\labs{\omega}\beta/2}$ is the reason why we had to truncate the frequency integral $\labs{\omega}\le \Omega.$ The second term ($\int_{\labs{\omega}>\Omega}$) is controlled in~\eqref{eq:greater_Omega}, leading to
\begin{align}
 \lnorm{\int_{-\epsilon}^{\epsilon}\int_{\labs{\omega}> \Omega}[\hat{\vA}(\omega,t),\vO] \rd \omega\rd t}_{\vrho}\lesssim\epsilon \e^{-\beta_0 \Omega} \frac{\e^{\sigma^2 \beta_0^2}}{\beta_0\sqrt{\sigma}}\L(\norm{\vrho_{\beta_0}\vA\vrho_{\beta_0}^{-1}}+\norm{\vrho_{\beta_0}^{-1}\vA\vrho_{\beta_0}}\R)\,. \label{eq:sqrtEprime}
\end{align}

\textbf{STEP 3: Optimize parameters.}  We collect the estimates 
\begin{align}
    \lnorm{[\vA,\vO]}_{\vrho} &\lesssim \epsilon^2 {d^2|A|}+ \frac{1}{c}\L( \e^{-\beta_0 \Omega} \frac{\e^{\sigma^2 \beta_0^2}}{\beta_0\sqrt{\sigma}}\L(\norm{\vrho_{\beta_0}\vA\vrho_{\beta_0}^{-1}}+\norm{\vrho_{\beta_0}^{-1}\vA\vrho_{\beta_0}}\R) + \frac{\e^{\sigma^2\beta^2/16}}{\sqrt{\epsilon}\sqrt{g(\epsilon)\beta}}\e^{\Omega\beta/4} \sqrt{\CE(\vO,\vO)}\R)\tag*{(By~\eqref{eq:eps3},\eqref{eq:sqrtE},\eqref{eq:sqrtEprime})}\\
    \Rightarrow  \lnorm{[\vA,\vO]}_{\vrho}&\lesssim \epsilon^2{d^2|A|} + \L(\frac{\e^{\sigma^2 \beta_0^2}}{\sigma\beta_0}+ \frac{\e^{\sigma^2\beta^2/16}}{\sqrt{g(1)\beta \sigma}} \R)\L(\norm{\vrho_{\beta_0}\vA\vrho_{\beta_0}^{-1}}+\norm{\vrho_{\beta_0}^{-1}\vA\vrho_{\beta_0}}\R)^{\frac{\beta}{\beta+4\beta_0}} \L(\frac{\sqrt{\CE(\vO,\vO)}}{\sqrt{\epsilon}} \R)^{\frac{4\beta_0}{\beta+4\beta_0}}\tag*{(Optimizing $\Omega$)}\\
    \Rightarrow  \lnorm{[\vA,\vO]}_{\vrho} &\lesssim \labs{A}d^2\L(\frac{\e^{\sigma^2 \beta_0^2}}{\beta_0\sigma}+ \frac{\e^{\sigma^2\beta^2/16}}{\sqrt{g(1)\beta\sigma}} \R)^{\frac{\beta+4\beta_0}{\beta+5\beta_0}} \L(\norm{\vrho_{\beta_0}\vA\vrho_{\beta_0}^{-1}}+\norm{\vrho_{\beta_0}^{-1}\vA\vrho_{\beta_0}}\R)^{\frac{\beta}{\beta+5\beta_0}}  \L(\sqrt{\CE(\vO,\vO)} \R)^{\frac{4\beta_0}{\beta+5\beta_0}} \tag*{(optimizing $\epsilon^2{c}+a/\epsilon^{b}\lesssim {c}\,a^{2/(2+b)}$,  {for $b/2c<1$})}.
\end{align}
The second line sets
\begin{align}
    \e^{\Omega} = \L(\sqrt{\epsilon}\,\,\frac{\norm{\vrho_{\beta_0}\vA\vrho_{\beta_0}^{-1}}+\norm{\vrho_{\beta_0}^{-1}\vA\vrho_{\beta_0}}}{\sqrt{\CE(\vO,\vO)} }\R)^{\frac{1}{\beta/4+\beta_0}},
\end{align}
and sets $\epsilon\le 1/(d\sqrt{|A|})\le 1$ otherwise the bound is vacuous.

\end{proof}

\section{Proof of main results}\label{sec:proofsthms}
We put the lemmas together for a streamlined proof of the main result.
\subsection{Proof of~\autoref{thm:main}}\label{sec:proofmainthm}
\begin{proof}[Proof of~\autoref{thm:main}]

We rewrite the nontrivial component on $A$ by commutators with nontrivial Pauli strings $\vS\in P_A$  
\begin{align}
    \vrho_\beta-\vrho_{\beta,-A}=
  \frac{1}{2^{2\labs{A}+1}} \sum_{\vS\in P_A} [\vS,[\vS,\vrho_\beta]].
\end{align}
Therefore, for any global operator $\vX$ such that $\norm{\vX}\le 1$, we have that 
\begin{align}
    \labs{\tr[\vX\CR_{A,t}[\vrho_\beta-\vrho_{\beta,-A}]]}&= \frac{1}{2^{2\labs{A}+1}} \labs{\sum_{\vS\in P_A}\tr\L[ \vrho_\beta [\vS,[\vS,\CR^{\dagger}_{A,t}[\vX]]]\R]}\tag*{($\tr[\vB[\vA,[\vA,\vC]]= \tr[\vC[\vA,[\vA,\vB]]$)}\\ 
&\le \frac{1}{2^{2\labs{A}+1}}\sum_{\vS\in P_A} \norm{\vI}_{\vrho_\beta}\cdot \norm{[\vS,[\vS,\CR^{\dagger}_{A,t}[\vX]]]}_{\vrho_\beta}\tag*{(Cauchy-Schwarz)}.
\end{align}
From now on, we denote by $r$ a coefficient which may depend on $\beta,\beta_0, d,\sigma$ and which may change from line to line. For the low-temperature case $\beta > 4\beta_0$, we begin with peeling off the outer-most product with $\vS=\prod_{j=1}^{w(\vS)}\vA^j$ of weight $w(\vS)$:
\begin{align}
    (\operatorname{RHS})&\lesssim \frac{1}{2^{2\labs{A}}}\sum_{\vS\in P_A} {2^{\labs{A}}}\,\L(r\norm{[\vS,\CR^{\dagger}_{A,t}[\vX]]}_{\vrho_\beta}\R)^{\frac{4\beta_0}{\beta}}\tag*{(\autoref{lem:weight_norm_comm},\autoref{cor:highweight})}\\
& \lesssim \frac{r}{2^{{\labs{A}}}}\sum_{\vS\in P_A} \L(2^{w(\vS)} \sum_{j=1}^{w(\vS)} \lnorm{[\vA^j,\CR^{\dagger}_{A,t}[\vX]]}_{\vrho_\beta}^{16\beta^2_0/\beta^2}
 \R)^{\frac{4\beta_0}{\beta}}\tag*{(\autoref{cor:global_local_comm})}\\
 & \lesssim \frac{r}{2^{\labs{A}}} 4^{\labs{A}} \L(\labs{A} 2^{\labs{A}}\sum_{a \in P^1_A }\lnorm{[\vA^a,\CR^{\dagger}_{A,t}[\vX]]}_{\vrho_\beta}^{16\beta^2_0/\beta^2}
 \R)^{\frac{4\beta_0}{\beta}}\tag*{($w(\vS)\le \labs{A}$)}
\end{align}
The first line uses that at $\beta_0 = 1/{4d}$, we have $\norm{\vrho_{\beta_0}\vA\vrho_{\beta_0}^{-1}},\norm{\vrho_{\beta_0}^{-1}\vA\vrho_{\beta_0}} \le 2^{\labs{A}}$ (\autoref{cor:highweight}). In the second and third lines, we reduce global jumps to local jumps. Next, we reduce commutators to Dirichlet forms.
\begin{align}
    (cont.)&\lesssim \labs{A} 2^{{(1+4\beta_0/\beta)}\labs{A}} r\L(\sum_{a \in P^1_A }\CE_a(\CR^{\dagger}_{A,t}[\vX])^{\frac{32\beta^3_0}{\beta^2(\beta+5\beta_0)}}
 \R)^{\frac{4\beta_0}{\beta}} \tag*{(\autoref{lem:commutator_Dirichlet})}\\ 
 & \lesssim \labs{A}^2 2^{2|A|} r \L(\frac{\sum_{a \in P^1_A } \CE_a(\CR^{\dagger}_{A,t}[\vX])}{3\labs{A}}\R)^{\frac{128\beta^4_0}{\beta^3(\beta+5\beta_0)}}\\
 &\le \labs{A}^2 2^{2|A|} r \L(\CE_A(\CR^{\dagger}_{A,t}[\vX])\R)^{\frac{128\beta^4_0}{\beta^3(\beta+5\beta_0)}}\\
&\lesssim r \cdot \labs{A}^{2}2^{2|A|}  \cdot t^{-\frac{128\beta^4_0}{\beta^3(\beta+5\beta_0)}}\tag*{(\autoref{cor:dirichlet})}.
\end{align}
The second line restore the normalization $1/3\labs{A}$ using that $\BE[\labs{x}^{\alpha}]\le \BE[\labs{x}]^{\alpha} $ for $0\le\alpha\le 1$. 

Next, at high temperatures $\beta \le 4\beta_0$, 
\begin{align}
    (\operatorname{RHS})
&\lesssim \frac{1}{2^{2\labs{A}+1}}\sum_{\vS\in P_A} 2^{\labs{A}}\norm{[\vS,\CR^{\dagger}_{A,t}[\vX]]}_{\vrho_\beta}\tag*{(\autoref{lem:loose_AO},\autoref{cor:highweight})}\\
&\lesssim\frac{1}{2^{|A|}}\, \sum_{\vS\in P_A}2^{w(\vS)}\sum_{j=1}^{w(\vS)}\|[\vA^j,\CR^{\dagger}_{A,t}[\vX]]\|_{\vrho_\beta}
\tag*{(\autoref{cor:global_local_comm})}\\
&\lesssim r \cdot \labs{A}^2 2^{2|A|}\cdot \L({\CE_A(\CR^{\dagger}_{A,t}[\vX])}\R)^{\frac{2\beta_0}{\beta+5\beta_0}}\tag*{(\autoref{lem:commutator_Dirichlet})}\\
&\lesssim r \cdot \labs{A}^22^{2\labs{A}} \cdot t^{-\frac{2\beta_0}{\beta+5\beta_0}}\tag*{(\autoref{cor:dirichlet})}
\end{align}
for some explicit constant $r$ depending on $d$, $\beta$ and $\sigma$. Optimizing over all operators $\vX$ such that $\norm{\vX}\le1$, using stationary property $\CR_{A,t}[\vrho_\beta] = \vrho_\beta$, and setting $\sigma = 1/\beta$, we conclude the proof.
\end{proof}
\subsection{Proof of~\autoref{cor:main_qlocal} \label{sec:pf_qlocal_decay}and~\autoref{cor:CMI}}
Let us derive the quasi-locality of the recovery map and decay of CMI using standard arguments.
\begin{proof}[Proof of~\autoref{cor:main_qlocal}]
For simplicity, we drop the inverse temperature $\beta$ in $\vrho_{\beta}\equiv \vrho$. By the main recovery guarantee of~\autoref{thm:main}, and quasi-locality estimates of~\autoref{lem:truncation} and \autoref{lem:quasilocal}, we have 
\begin{align*}
2\Delta:=\lnorm{\vrho-\mathcal{R}_{A,t,\ell}[\vec{\tau}_A \otimes \vrho_{BC}]}_1 &\lesssim \left\| \vrho-\mathcal{R}_{A,t}[\vrho_{-A}]\right\|_1+ \left\| (\mathcal{R}_{A,t}-\mathcal{R}_{A,t,\ell})[\vrho_{-A}]\right\|_1\\
&\lesssim r\e^{\mu|A|}t^{-\lambda} + t|A|\left(\e^{-c'\frac{\ell}{d\beta}}+2^{-\ell}\right)
\end{align*} 
for every $\ell \ge 4\e^2\beta d$ and some absolute constant $c'$. From this, it suffices to choose 
\begin{align}
 t = \e^{(\mu \labs{A}+m\ell)/{\lambda+1}}
 \quad \text{with}\quad m=\min\L(\ln(2),\frac{c'}{d\beta}\R)
\end{align}
so that
\begin{align*}
\Delta\lesssim r|A| \exp\L(\frac{\mu\labs{A} -m\lambda\ell}{1+\lambda} \R) \le r|A| \exp\L(\frac{\mu\labs{A} -m\lambda \ell}{2} \R)\le r \exp\L(\frac{(\mu+2)\labs{A} -m\lambda \ell}{2} \R) \,
\end{align*}
where the second inequality assumes that $0<\lambda<1$ and uses that the bound is vacuous unless the exponent is negative. 
\end{proof}

\begin{proof}[Proof of~\autoref{cor:CMI}]
Let $\Lambda=A\sqcup B\sqcup C$ be a partitioning of the system with $A$ shielded from $C$. Consider the recovery map $\mathcal{R}_{A,t,\ell}$ associated with the Hamiltonian $\vH_\ell$, where $\ell=\operatorname{dist}(A,C)-1$ is chosen so that $\vH_\ell$ is supported on $B$ and does not overlap with $C$. Then, the CMI can be bounded by the approximate recovery map~\cite[Theorem 11.10.5]{wilde2011classical}
\begin{align}
I(A:C|B)_{\vrho} \le \Delta \log(\operatorname{dim}({C})) + h_2(\Delta) \lesssim \log(\operatorname{dim}({C}))\sqrt{\Delta} 
\end{align}
where $h_2(x) =  - x\log_2(x) - (1-x)\log_2(1-x)\lesssim \sqrt{x}$ for $x \in [0,1]$ and
\begin{align}
\Delta:= \frac{1}{2}\lnorm{\vrho-\mathcal{R}_{A,t,\ell}[\vec{\tau}_A \otimes \vrho_{BC}]}_1.    
\end{align}
Plug the bounds for $\Delta$ from~\autoref{cor:main_qlocal} and update $\mu',\lambda', r'$ to conclude the proof. 
\end{proof}

\subsection{Quasi-local Gibbs preparation guarantee}
This section follows closely the patching argument of~\cite{brandao2019finite_prepare} and upgrades their nonconstructive recovery map with our time-averaged map. This removes the uniform Markov asssumption, but our map has worse locality (due to the $\e^{\mu \labs{A}}$ prefactor ), leading to a polylogarithmic circuit depth overhead. 

\begin{figure}[h!]
\includegraphics[width=0.8\textwidth]{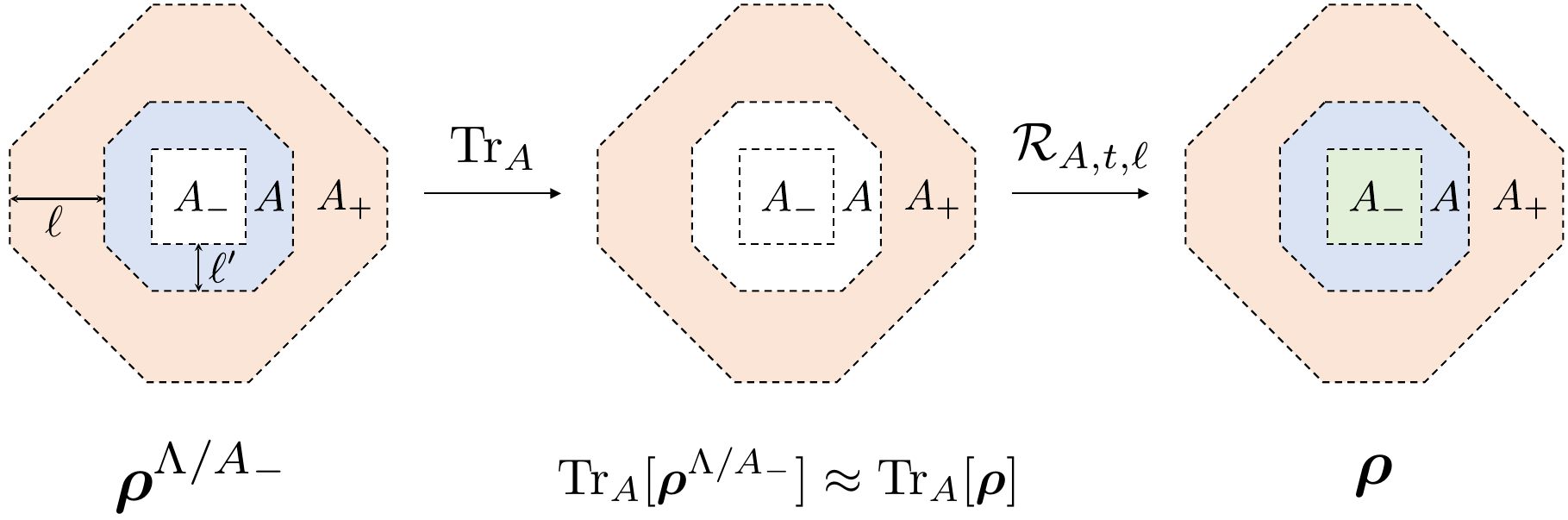}\caption{Combining local indistinguishability (\autoref{thm:localindistinguishability}) and local Markov property (\autoref{thm:main}) to recover the Gibbs state from a restricted Hamiltonian. There are two length scales: $\ell'$ as the correlation length for local indistinguishability and $\ell$ as the quasi-locality of the recovery map.
}\label{fig:zoomin}
\end{figure}

\begin{proof}[Proof of~\autoref{cor:samplingresult}]

For simplicity, we drop once again the inverse temperature $\beta$ in $\vrho_{\beta}\equiv \vrho.$ We restrict ourselves to the two dimensional setting $D=2$, since the proof readily extends to higher dimensions: we consider a tiling of the lattice with patches of three types: 
\begin{align}
A_-^{h,j}\subset A^{h,j}\subset A^{h,j}_+,\quad &\text{for}\quad j=1,\dots, N_A\quad \text{and}\quad  \,h=1\dots,h_0
\end{align} 
where $A^{h,j}_+$ denotes the region of sites at distance at most $\ell$ away from $A^{h,j}$, and such that $\Lambda=\bigsqcup_{h,j}A^{h,j}_-$ and $A^{h,j}_+\cap A_+^{h,k}=\emptyset$ for $j\ne k$ (see Figure \ref{figlattice}). Next, we denote the quasi-local maps (trace-out-and-recovery)
\begin{align}
\mathbb{F}_{A^{h,j}}:=\mathcal{R}_{A^{h,j},t,\ell}\circ (\vec{\tau}_{A^{h,j}} \otimes \tr_{A^{h,j}})\quad \text{where}\quad \mathcal{R}_{A^{h,j},t,\ell}:=\frac{1}{t}\int_0^{t} \operatorname{exp}\Big(s \mathcal{L}_{A,\ell}\Big)\,\rd s
\end{align}
with $\mathcal{L}_{A,\ell}$ the generator corresponding to single qubit jumps $\vA^a\subset P^1_A$ and Gibbs state $\vrho_{A^{h,j}_+}$. That is, $\mathcal{R}_{A^{h,j},t,\ell}$ acts only on $A^{h,j}_+.$
\smallskip 

\noindent \textbf{First step.} 
By~\autoref{cor:main_qlocal}, we have that there is a good time $t^*$ depending on $\ell$ such that

\begin{align*}
\|\mathbb{F}_{A^{1,j}}[\vrho]-\vrho\|_1=\|\mathcal{R}_{A^{1,j},t^*,\ell}[\vrho_{-A^{1,j}}]-\vrho\|_1
\lesssim r|\exp\L(\frac{\mu'\labs{A^{1,j}} -m\lambda \ell}{2} \R)\equiv  \Delta(\labs{A^{1,j}},\ell).
\end{align*}
Hence, denoting $\mathbb{F}_{A^1}=\bigotimes_{j=1}^{N_A} \mathbb{F}_{A^{1,j}}$ (Note that the channels $\mathbb{F}_{A^{1,j}}$ act on disjoint regions), applying triangle inequality with the telescoping sum yields
\begin{align*}
\lnorm{\mathbb{F}_{A^1}[\vrho]-\vrho}_1 = \lnorm{\bigotimes_{j=1}^{N_A} \mathbb{F}_{A^{1,j}}[\vrho]-\vrho}_1  
\lesssim N_A\,\Delta(\labs{A^{1,j}},\ell)\,.
\end{align*}
Next, by \autoref{thm:localindistinguishability}, for any state $\vsigma_{A_-}$ on $A^1_-=\bigsqcup_{j=1}^{N_A}A_-^{1,j}$,
\begin{align}
\|\tr_{A^1}[\vrho-\vrho^{\Lambda\backslash A^1_-}\otimes\vsigma_{A^{1}_-}]\|_1\lesssim N_A\,\poly(\labs{A},\ell'^{D})\,|\partial A_-^{\max}|\,\e^{-\frac{\ell'}{\xi'}}
\end{align}
for some modified correlation length $\xi'>0$, $\ell'\le\operatorname{dist}(A^1_-,A^1_+)$, and where $|\partial A_-^{\max}|$ stands for the maximal size that any boundary $A^{h,j}_-$ can take. Thus, combining the last two equations, we arrive at
\begin{align*}
\|\mathbb{F}_{A^1}[\vrho^{\Lambda\backslash A^1_-}\otimes \vsigma_{A^1_-}]-\vrho\|_1\lesssim N_A\,\poly(\labs{A},\ell'^{D})|\partial A_-^{\max}|\,\e^{-\frac{\ell'}{\xi'}}+N_A\Delta(|A^{1,j}|,\ell)\,.
\end{align*}
\begin{figure}[h!]
\includegraphics[width=0.9\textwidth]{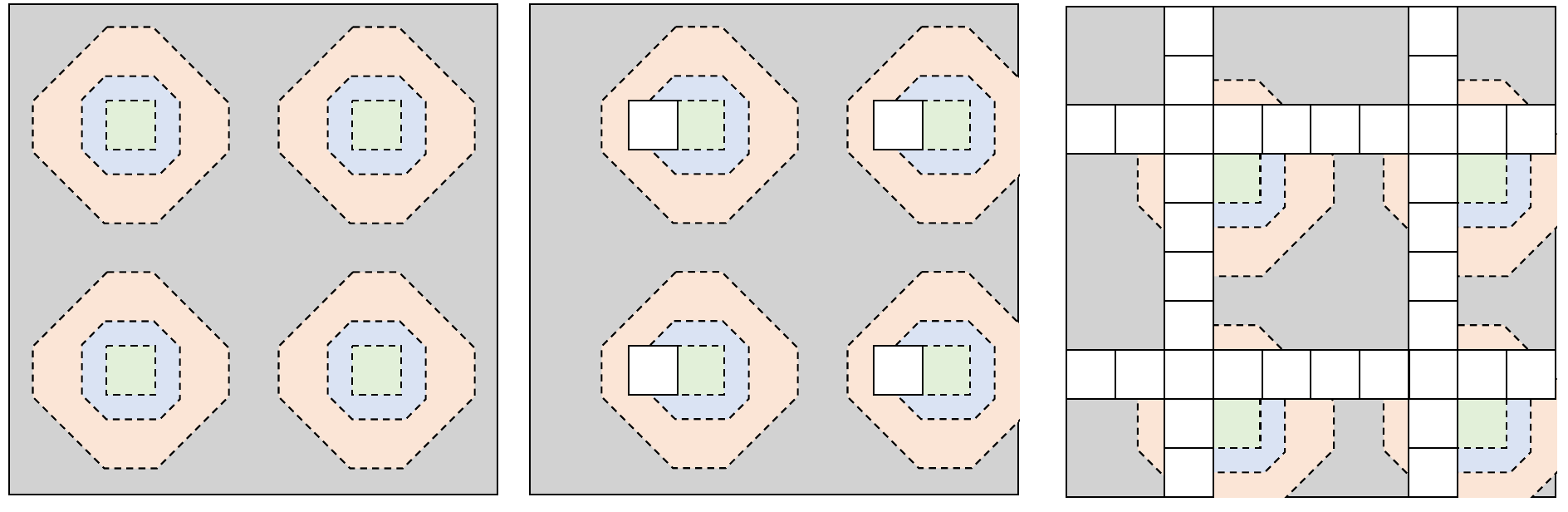}\caption{Building the global Gibbs states from quasi-local patches (see~\autoref{fig:zoomin} for the specification of each patch) in parallel. The patching argument proceeds by punching point-like holes (white squares) until the whole lattice is covered. Note that we assumed local indistinguishability holds for all intermediate Hamiltonians (supported on colored regions), which may have lots of punched holes and a nontrivial topology.
}\label{figlattice}
\end{figure}
\textbf{Repeat.}
Next, we repeat the previous scheme with the regions of the $A^{h}$ type, $h=2,\dots, h_0$: at each round, defining the quantum channel $\mathbb{F}_{A^h}:=\bigotimes_{j=1}^{N_{A}}\mathbb{F}_{A^{h,j}}$, as well as $A_-^{[h]}=\sqcup_{h'=1}^h A_-^{h'}$, we get that for any state $\vsigma_{A^{[h]}}$, 
\begin{align*}
\|\mathbb{F}_{A^{h}}[\vrho^{\Lambda\backslash A_-^{[h]}}\otimes \vsigma_{A_-^{[h]}}]-\vrho^{\Lambda\backslash A^{[h-1]}_-}\otimes \vsigma_{A^{[h-1]}_-}\|_1\lesssim N_{A}\,\poly(\labs{A},\ell'^{D})|\partial A_-^{\max}|\,\e^{-\frac{\ell'}{\xi'}}+N_{A}\Delta(|A^{1,j}|,\ell)\,,
\end{align*}
where once again we chose $\ell'=\operatorname{dist}(A^{h}_-,A^{h}_+)$.
Combining with all the previous rounds, we hence get that, for $\mathcal{C}_n\equiv \mathbb{F}_{A^{[h]}}=\bigcirc_{h'=1}^h \mathbb{F}_{A^{h'}}$ (note that the $\mathbb{F}_{A^{h'}}$'s overlap with each other so may not commute), 
\begin{align*}
\|\mathbb{F}_{A^{[h]}}[\vrho^{\Lambda\backslash A^{[h]}_-}\otimes \vsigma_{A^{[h]}_-}]-\vrho\|_1\lesssim h\,N_A\poly(\labs{A},\ell'^{D})|\partial A_-^{\max}|\e^{-\frac{\ell'}{\xi'}}+h\, N_A\Delta(|A^{1,j}|,\ell)\,.
\end{align*}
In the final step $h=h_0$, we hence have that for any $\vsigma_{\Lambda}$, 
\begin{align*}
\|\mathbb{F}_{A^{[h_0]}}[ \vsigma_{\Lambda}]-\vrho\|_1\lesssim  h_0\,N_A\poly(\labs{A},\ell'^{D})|\partial A_-^{\max}|\e^{-\frac{\ell'}{\xi'}}+h_0\, N_A\Delta(|A^{1,j}|,\ell)\,.
\end{align*}

\textbf{Altogether.}
Now, since $N_A\le |\Lambda|=n$, we can choose $A_-^{h,j}$ of side-length $\log(n)$ so that $|\partial A_-^{\max}|\lesssim \log(n)$. Then choosing $\ell=\Theta(\log^2(n/\epsilon))$, $t^*=\e^{\Theta(\log^2(n{/\epsilon}))}$, and $\ell'=\mathcal{O}(\log(n))$
ensures that $\|\mathbb{F}_{A^{[h_0]}}(\vsigma_\Lambda)-\vrho\|_1\le \epsilon$. Finally, we use \autoref{circuitimplementation} to control the number of gates needed to implement the map $\mathbb{F}_{A^{[h_0]}}$ within the circuit model. Similarly, for dimension $D$, the parameters choices would then be $t^*=\e^{\Theta(\log^D(n{/\epsilon}))}$, $\ell=\Theta(\log^D(n/\epsilon))$ and $\ell'=\mathcal{O}(\log(n)).$

\end{proof}

\appendix
\section{Proof of quasi-locality (Lemma \ref{lem:quasilocal})}
\label{sec:LRestimates}
In what follows, we localize the detailed balanced Lindbladian~\eqref{eq:exact_DB_L} by truncating the Hamiltonian. We split the Lindbladian into the coherent part and the dissipative part
\begin{align}
    \CL = -\ri[\vB,\cdot] + \CD\label{eqdecompdissip}
\end{align}
which will be treated slightly differently.

\subsection{Review of Lieb-Robinson bounds}
We instantiate the most standard Lieb-Robinson argument \cite{Lieb1972,haah2020quantum, chen2023speed}:
\begin{lem}[Local Hamiltonian patch]\label{lem:LRbounds}
For a Hamiltonian $\vH = \sum_{\gamma}\vH_{\gamma}$ with bounded interaction degree $d$ (\autoref{sec:Ham}) and an operator $\vA$ supported on region $A\subset \Lambda$, let $\vH_{\ell}$ contain all terms $\vH_{\gamma}$ such that $\operatorname{dist}(\gamma,A) < \ell-1$ for an integer $\ell$. Then,
\begin{align}
\norm{\e^{i\vH_{\ell} t}\vA \e^{-i\vH_{\ell} t}-\e^{i\vH t}\vA \e^{-i\vH t}} \lesssim \norm{\vA} \labs{A}\frac{(2d\labs{t})^{\ell}}{\ell!}.
\end{align}    
\end{lem}
 We largely reproduce~\cite[Proposition 4.3]{chen2023speed} as follows.
\begin{proof}
Without loss of generality, let $\norm{\vA}\le1,$ and apply~\cite[Proposition 4.3]{chen2023speed}, where the number of self-avoiding paths of length $p\ge 1$ is bounded by $\labs{A}d\cdot (d-1)^{p-1}\le \labs{A}d^{p}$
\begin{align}
    \norm{\e^{i\vH_{\ell} t}\vA \e^{-i\vH_{\ell} t}-\e^{i\vH t}\vA \e^{-i\vH t}} &\le  \labs{A}\sum_{p=\ell}^{\infty} \frac{(2d\labs{t})^p}{p!}\\
    &\le \labs{A}\frac{(2d\labs{t})^{\ell}}{\ell!} \sum_{q=0}^{\infty} \frac{(2d\labs{t})^q}{(q+\ell)!},\\
    \operatorname{LHS}&\le \min\L(2,\labs{A}\frac{(2d\labs{t})^{\ell}}{\ell!} \frac{1}{1-2/e}\R). 
\end{align}
The last bound uses that the RHS is meaningful only if $2d\labs{t}/r\le 2$ and sums the decaying series.
\end{proof}

\subsection{Dissipative part}
Now, we apply the Lieb-Robinson bounds to the Heisenberg dynamics in the Lindbladian. This section treats the dissipative part as follows, with some generality of the transition weight.
\begin{lem}[Quasi-locality of the dissipative part]
Consider the Lindbladian~\eqref{eq:exact_DB_L} with any transition weight $0\le \gamma(\omega)\le 1$ for a single jump operator $\vA$. Then, the dissipative part can be approximated by
\begin{align}\label{lem.dissipaticepartboundl}
    \norm{\CD-\CD_{\ell}}_{\diamond}\lesssim \norm{\vA} \left( \e^{-\sigma^2\ell^2/2d^2}+\labs{A}\e^{-\ell}\right)
\end{align} 
where $\CL_{\ell}$ is defined by replacing $\vH$ with the localized Hamiltonian $\vH_{\ell}$ as in~\autoref{lem:LRbounds}.
\end{lem}
\begin{proof}
Without loss of generality, let $\norm{\vA}\le1.$ It will be helpful to consider the ``purified'' jump (which is an isometry if $\vA^{\dagger}\vA=\vI$): denoting by $\hat{\vA}_{\vH}$, resp. by $\hat{\vA}_{\vH_\ell}$, the operator Fourier transform of ${\vA}$ associated to Hamiltonian $\vH$, resp. associated to $\vH_\ell$,
\begin{align}
    \vV := \int_{-\infty}^{\infty} \hat{\vA}_{\vH}(\omega) \otimes \ket{\omega}\rd \omega\quad \text{and}\quad \vV' := \int_{-\infty}^{\infty} \hat{\vA}_{\vH_{\ell}}(\omega) \otimes \ket{\omega}\rd \omega,
\end{align}
and the filter as an operator
\begin{align}
    \vF = \int_{-\infty}^{\infty} \gamma(\omega) \vI\otimes \ket{\omega}\bra{\omega} \rd \omega
\end{align}
where $\braket{\omega'|\omega} = \delta(\omega'-\omega)$\footnote{More formally, $\vV$ can be defined as the map $\mathbb{C}^{2^n}\to \mathbb{C}^{2^n}\otimes L^2(\mathbb{R})$, $|\psi\rangle\mapsto (\omega\mapsto \vA_{\vH}(\omega)|\psi\rangle)$, with norm $\|\vV\|:=\sup_{\||\psi\rangle\|\le 1}\sqrt{\int_{\mathbb{R}}\|\hat{\vA}_{\vH}(\omega)|\psi\rangle\|^2\,\rd\omega}$. Whenever $\vA^\dagger \vA=\vI$, the previous integral simply evaluates to $\||\psi\rangle\|^2$, giving $\|\vV\|=1$.
Furthermore, the map $\vF$ is defined as the point-wise multiplication by the function $\gamma\in L^\infty(\mathbb{R})$ on $L^2(\mathbb{R})$, so that $\|\vF\|\le \|\gamma\|_\infty$. Above, we also denote by $\tr_\omega$ the integral over $\omega\in\mathbb{R}$.}. 
Then, for any input $\vrho$ (which may be entangled with ancillas)
\begin{align}
\CD[\vrho] &= \tr_{\omega}\Big[ \vF\vV \vrho \vV^{\dagger} - \frac{1}{2}\{\vV^{\dagger}\vF\vV,\vrho\}\Big] \\
\CD_{\ell}[\vrho] &= \tr_{\omega}\Big[ \vF\vV' \vrho \vV^{'\dagger} - \frac{1}{2}\{\vV^{'\dagger}\vF\vV',\vrho\}\Big],
\end{align}
and clearly
\begin{align}
    \norm{\CD-\CD_{\ell}}_{\diamond} &\le 2 \norm{\vV-\vV'}(\norm{\vV}+\norm{\vV'})\norm{\vF}\\
    &\le 4\norm{\vV-\vV'}    
\end{align}
since $\norm{\vF}\le \sup_{\omega}\labs{\gamma(\omega)}\le 1$ and $\norm{\vV},\norm{\vV'}\le \norm{\vA} \le 1.$
Now we compute the difference
\begin{align}
    \vV = \int_{-\infty}^{\infty} \hat{\vA}_{\vH}(\omega) \otimes\ket{\omega}\rd \omega
&= \frac{1}{\sqrt{2\pi}}\int_{-\infty}^{\infty}\int_{-\infty}^{\infty} 
f(t) \e^{-i\omega t} 
\e^{i\vH t}\vA \e^{-i\vH t} \rd t \otimes \ket{\omega}\rd \omega\\
&= \int_{-\infty}^{\infty} 
f(t)
\e^{i\vH t}\vA \e^{-i\vH t}  \otimes \vU_{FT}\ket{t} \rd t\\
&\approx \int_{-T}^{T} 
f(t)
\e^{i\vH t}\vA \e^{-i\vH t}  \otimes \vU_{FT}\ket{t} \rd t\\
&\approx \int_{-T}^{T} 
f(t)
\e^{i\vH_{\ell} t}\vA \e^{-i\vH_{\ell} t}  \otimes \vU_{FT}\ket{t} \rd t\\
&\approx \int_{-\infty}^{\infty} 
f(t)
\e^{i\vH_{\ell} t}\vA \e^{-i\vH_{\ell} t}  \otimes \vU_{FT}\ket{t} \rd t
 = \vV'.
\end{align}
with the Fourier transform unitary on $L^2(\mathbb{R})$, written as $\vU_{FT} := \frac{1}{\sqrt{2\pi}}\int_{-\infty}^{\infty}\int_{-\infty}^{\infty}\e^{-i \omega t}\ket{\omega}\bra{t}\rd \omega \rd t$ in bracket notations.
It remains to quantify the errors made in the last three approximations:
\begin{align}
    \norm{\vV-\vV'} \le 2\norm{f(t)\indicator(\labs{t}\ge T)}_2 + \sup_{\labs{t}\le T}\lnorm{\e^{i\vH_{\ell} t}\vA \e^{-i\vH_{\ell} t}-\e^{i\vH t}\vA \e^{-i\vH t}}
\end{align}    
where we used that $\|f\|_2=1$. We bound
\begin{align}
    \norm{f(t)\indicator(\labs{t}\ge T)}^2_2 \lesssim \int_{T}^{\infty} \e^{-2\sigma^2t^2}\sigma \rd t \lesssim \int_{\sqrt{2}\sigma T}^{\infty} \e^{-x^2}\rd x \lesssim \frac{\e^{-2\sigma^2 T^2}}{\sigma T }.
\end{align}
Since the LHS is bounded by one, we can further simplify 
\begin{align}
    \norm{f(t)\indicator(\labs{t}\ge T)}^2_2 \lesssim \min(1, \e^{-2\sigma^2 T^2}).
\end{align}
Also, by~\autoref{lem:LRbounds},
\begin{align}
    \norm{\e^{i\vH_{\ell} t}\vA \e^{-i\vH_{\ell} t}-\e^{i\vH t}\vA \e^{-i\vH t}} \lesssim \labs{A}\frac{(2dT)^{\ell}}{\ell!} \lesssim \labs{A} \Big(\frac{2ed T}{\ell}\Big)^{\ell},
\end{align}
using Stirling's approximation $1/\ell! \lesssim (e/\ell)^{\ell}$. 
Let us choose a convenient (slightly suboptimal)
\begin{align}
    T = \frac{\ell}{2d\e^2}
\end{align}
to arrive at the advertised bounds. 
\end{proof}

\subsection{The coherent part}
Recall the coherent term~\cite{chen2023efficient} for the Metropolis weight~\eqref{eq:Metropolis} associated with each jump $\vA^a$
\begin{align}\label{eq:coherent}
\vB^a:=\int_{-\infty}^\infty \,b_1(t)\, \,\e^{-i\beta \vH t}\lim_{\eta\rightarrow 0_+}\left(\int_{-\infty}^\infty\,b^\eta_2(t')\,{\vA}_{\vH}^{a}(\beta t'){\vA}^a_{\vH}(-\beta t')dt'+\frac{1}{{8}\sqrt{2}\,\pi}\,\vA^{a\dagger}\vA^{a}\right)\e^{i\beta \vH t} dt
\end{align}
where we denote $\vA_{\vH}(t):=\e^{i\vH t}\vA \e^{-i\vH t}$, 
\begin{align}
&b_1(t):=2\sqrt{\pi}\,\e^{\frac{1}{8}}\,\left(\frac{1}{\operatorname{cosh}(2\pi t)}\ast \sin(-t)\operatorname{exp}(-2t^2)\right)\qquad \text{ with }\quad \|b_1\|_{1}<1\label{eq:b1}\\
&b^\eta_2(t):=\mathsf{1}(|t|\ge \eta)\frac{1}{{2}\sqrt{2}\,\pi}\,\frac{\operatorname{exp}(-2t^2-it)}{t(2t+i)}\equiv \mathsf{1}(|t|\ge \eta)\,b_2^M(t)
\end{align}
where we take the $\eta\rightarrow\infty$ limit in a slightly more convenient way than~\cite{chen2023efficient}.
The most generic bound from~\cite{chen2023efficient} has a logarithmic dependence on $\beta\norm{\vH}$, which depends on the global system size. Here, we want to use the locality of jump $A$ to obtain a norm bound on $\vB^{a}$ that is \textit{independent} of the system size. While this is not exactly what we want, this first generic argument will teach us how to truncate the Hamiltonian.
\begin{cor}[Bounds on the coherent term for local jumps]\label{cor:boundsCoherentLocal}
Suppose that $\nrm{\sum_{a\in P_A}\vA^{a\dagg}\vA^a}\le 1.$ Then,	
    \begin{align}
		\nrm{\sum_{a\in P_A} \vB^a}\leq \frac{1}{\sqrt{2}\pi}\L({\frac{9}{8}}+\frac{\sqrt{2\pi}}{{4}} + {2} \beta\lnorm{\sum_{a\in P_A} [\vA^{a\dagger},\vH][\vA^{a},\vH]}^{1/2}\R).
	\end{align}
\end{cor}

\begin{proof}
Let us focus on the seemingly divergent term in the inner integral~\eqref{eq:coherent}, and further split it into (similarly to~\cite{gilyen2024quantum})
\begin{align}
    \indicator(\labs{t'} \ge\eta) = \indicator(1\ge \labs{t'} \ge\eta)+\indicator(\labs{t'}> 1).
\end{align}
For ease of notations, we simply denote $\vA_{\vH}(t)=\vA(t)$. For the first interval,
\begin{align}
    &\lim_{\eta\rightarrow 0_+} \sum_{a\in P_A} \int_{-\infty}^{\infty} \indicator(1\ge \labs{t'} \ge\eta) b_2^{M}(t')\vA^{a\dagger}(\beta t')\vA^a(-\beta t')\rd t' \\
    &=  \lim_{\eta\rightarrow 0_+}\sum_{a\in P_A} \int_{-\infty}^{\infty} \frac{\indicator(1\ge \labs{t'} \ge\eta)}{t'} t'b_2^{M}(t')\vA^{a\dagger}(\beta t')\vA^a(-\beta t')\rd t' \\
    & = \lim_{\eta\rightarrow 0_+} \sum_{a\in P_A} \int_{-\infty}^{\infty} \frac{\indicator(1\ge \labs{t'} \ge\eta)}{t'} \int_0^{t'} \frac{\rd}{\rd t'} \L(t'b_2^{M}(t')\vA^{a\dagger}(\beta t')\vA^a(-\beta t')\R)\rd t' \tag*{(Fundamental theorem of Calculus)}
\end{align}
where we use the symmetry of the integral to remove the leading order expansion in $t',$ hence explicitly removing the singularity. Note that we do not need a convergence of Taylor series for $t'b_2^{M}(t')\vA^{a\dagger}(\beta t')\vA^a(-\beta t')$, but only that it is differentiable on the interval $[0,1]$. Therefore,
\begin{align}
    &\lnorm{\lim_{\eta\rightarrow 0_+} \sum_{a\in P_A} \int_{-\infty}^{\infty} \frac{\indicator(1\ge \labs{t'} \ge\eta)}{t'} \int_0^{t'} \frac{\rd}{\rd t'} \L(t'b_2^{M}(t')\vA^{a\dagger}(\beta t')\vA^a(-\beta t')\R)\rd t'} \\
    &\le \int_{-1}^{1} \rd t' \cdot   \sup_{\labs{t'}\le 1}\lnorm{\frac{\rd}{\rd t'} \L(t'b_2^{M}(t')\sum_{a\in P_A}\vA^{a\dagger}(\beta t')\vA^a(-\beta t')\R)}.
\end{align}

Therefore, it remains to bound the derivative
\begin{align}
    \lnorm{\sum_{a\in P_A} \frac{\rd}{\rd t'} \L(t'b_2^{M}(t')\vA^{a\dagger}(\beta t')\vA^a(-\beta t')\R)} &\le \lnorm{\sum_{a\in P_A} \frac{\rd}{\rd t'} \L(t'b_2^{M}(t')\R) \vA^{a\dagger}(\beta t')\vA^a(-\beta t')} \\
    &+ \lnorm{\sum_{a\in P_A}  \L(t'b_2^{M}(t')\R) \beta [\vA^{a\dagger},\vH](\beta t')\vA^a(-\beta t')}\\
    &+ \lnorm{\sum_{a\in P_A}  \L(t'b_2^{M}(t')\R) \vA^{a\dagger}(\beta t')\beta[\vA^a,\vH](-\beta t')}\label{eq:chainrule}\\
    &\le \labs{ \frac{\rd}{\rd t'} \L(t'b_2^{M}(t')\R)} \lnorm{\sum_{a\in P_A} \vA^{a\dagger}\vA^{a}}\\
    &+ 2\beta\labs{t'b_2^{M}(t')} \cdot\lnorm{\sum_{a\in P_A} \vA^{a\dagger}\vA^{a}}^{1/2}\lnorm{\sum_{a\in P_A} [\vA^{a\dagger},\vH][\vA^{a},\vH]}^{1/2}
\end{align}
where we also used standard norm inequalities (see e.g. \cite[Lemma K.1]{chen2024local}). Next, we evaluate the scalar bounds
\begin{align}
    \labs{ \frac{\rd}{\rd t} \L(tb_2^{M}(t)\R)} &\le \frac{1}{{2}\sqrt{2}\pi}\labs{-(4t+i)\frac{\exp(-2t^2-it)}{2t+i} -2\frac{\exp(-2t^2-it)}{(2t+i)^2}} \le \frac{\e^{-2t^2}}{{2}\sqrt{2}\pi}\le \frac{1}{{2}\sqrt{2}\pi}\\
    \labs{tb_2^{M}(t)}& \le \frac{\e^{-2t^2}}{{2}\sqrt{2}\pi}\le \frac{1}{{2}\sqrt{2}\pi}.
\end{align}
For the second interval,
\begin{align}
    \lnorm{ \sum_{a\in P_A} \int_{-\infty}^{\infty} \indicator(\labs{t'}>1) b_2^{M}(t') \vA^{a\dagger}(\beta t')\vA^a(-\beta t')\rd t'} &\le \lnorm{\sum_{a\in P_A} \vA^{a\dagger}\vA^{a}}\,\, \lnorm{t'\mapsto \indicator(\labs{t'}>1) b_2^{M}(t')}_1\\
    &\le  \frac{1}{{2}\sqrt{2}\pi} \lnorm{\exp\L(-2t^2\R)}_1 \le \frac{1}{{4}\sqrt{\pi}}\,.
\end{align}
Collect the constants and the contribution from $\frac{1}{8\sqrt{2}\,\pi}\,\vA^{a\dagger}\vA^{a}$ to conclude the proof.
\end{proof}
Based on the above way to tame the log-singularity, we obtain a local approximation of the coherent term without reference to the global system size.
\begin{cor}\label{cor:localjumps}
For a jump $\vA^a$ supported on region $A$ with $\norm{\vA^a}\le1$ and integer $\ell \ge 4 \e^2d\beta$, 
    \begin{align}
		\nrm{\vB^a-\vB^a_{\ell}}\lesssim \e^{-c'\frac{\ell}{d\beta}} +\labs{A} 2^{-\ell}\,,
	\end{align}
{for some absolute constant $c'>0$.}
\end{cor}
\begin{proof}
We drop the superscript $a$ for ease of notation. Truncating the double integration range, we get
\begin{align}\label{eq:MetropolisDimLessTime}
  \vB  &\!=\! \int_{-\infty}^{\infty}b_1(t)\e^{-\ri\beta \vH t}\left( \lim_{\eta \rightarrow 0}\int_{-\infty}^{\infty}b_2^{\eta}(t')\vA^{\dagger}(\beta t')\vA(-\beta t')\rd t'+\frac{1}{16\sqrt{2}\pi}\vA^{\dagger}\vA\right) \e^{\ri \beta\vH t}\rd t\\
    &\! \approx\! \int_{-T}^{T} b_1(t)\e^{-\ri\beta \vH t}\left( \lim_{\eta \rightarrow 0}\int_{-\infty}^{\infty}b_2^{\eta}(t')\vA^{\dagger}(\beta t')\vA(-\beta t')\rd t'+\frac{1}{16\sqrt{2}\pi}\vA^{\dagger}\vA\right) \e^{\ri \beta\vH t}\rd t\\
    &\! \approx\! \int_{-T}^{T} b_1(t)\e^{-\ri\beta \vH t}\left( \lim_{\eta \rightarrow 0}\int_{-T}^{T}b_2^{\eta}(t') \vA^{\dagger}(\beta t')\vA(-\beta t')\rd t'+\frac{1}{16\sqrt{2}\pi}\vA^{\dagger}\vA\right) \e^{\ri \beta\vH t}\rd t\\
    & \!=\!  \int_{-T}^{T}\!\!\! \!b_1(t)\left( \!\lim_{\eta \rightarrow 0}\int_{-T}^{T}\!\!\! \L( \indicator(\labs{t}\ge1) \!+\!\indicator(1\ge \labs{t}\ge \eta) \R) b_2^{\eta}(t') \vA^{\dagger}(\beta (t'-t))\vA(-\beta (t'+t))\rd t'\!+\frac{(\vA^{\dagger}\vA)(-\beta t)}{16\sqrt{2}\pi}\!\right)\rd t.
\end{align}
We then compare the above with that of the localized Hamiltonian $\vH\rightarrow \vH_{\ell}.$ We recall Lieb-Robinson bounds (\autoref{lem:LRbounds}) for various operators and times, and using that $\norm{b_1}_1,\norm{b_2\indicator(\labs{t}\ge 1)}_1 \lesssim 1$
\begin{align}
(\vA^{\dagger}\vA\quad \text{term})&\lesssim \labs{A} \frac{(2d\beta T)^{\ell}}{\ell!}\\
   (\indicator(\labs{t}\ge1)\quad \text{term}) &\lesssim \labs{A} \frac{(4d\beta T)^{\ell}}{\ell!}.
\end{align}
For the term $\indicator(1\ge \labs{t}\ge \eta)$, due to the $1/t$ divergence in $b_2^M(t),$ we have to revisit the proof of~\autoref{cor:boundsCoherentLocal}, and particularly compare $\vH$ vs. $\vH_{\ell}$ at~\eqref{eq:chainrule} to get
\begin{align}
      ( \indicator(1\ge \labs{t}\ge \eta)\quad \text{term}) &\lesssim \beta \sup_{\labs{t}\le \beta(T+1)}\lnorm{[\vA,\vH]_{\vH}(t)-[\vA,\vH]_{\vH_{\ell}}(t)}+\sup_{\labs{t}\le \beta(T+1)} \lnorm{\vA_{\vH}(t)-\vA_{\vH_{\ell}}(t)}\\
       &\lesssim \beta d\labs{A} \frac{(2d\beta (T+1))^{\ell-1}}{(\ell-1)!} +\labs{A}\frac{(2d\beta (T+1))^{\ell}}{\ell!}\\
       &\lesssim \labs{A}  \ell\frac{(2d\beta (T+1))^{\ell}}{\ell!} \tag*{(Simplification and $T\ge0$)}
\end{align}
where we used the notation $\vO_{\vH}(t)$, resp. $\vO_{\vH_\ell}(t)$, to indicate that the time evolution generated by $\vH$ resp. by $\vH_\ell$. We may use $\labs{A}\ell(4d\beta (T+1))^{\ell}/\ell! \le \labs{A}\ell(\frac{4ed\beta (T+1)}{\ell})^{\ell}$ to upper bound all the above truncation errors and arrive at
\begin{align}
    \nrm{\vB-\vB_{\ell}} &\lesssim \lnorm{b_1(t)\indicator(\labs{t}\ge T)}_1 + \lnorm{b^{M}_2(t)\indicator(\labs{t}\ge T)}_1 + \labs{A}\ell\left(\frac{4ed\beta (T+1)}{\ell}\right)^{\ell}\\
    &\lesssim \e^{-cT} + \e^{-2T^2} + \labs{A}\ell\left(\frac{4ed\beta (T+1)}{\ell}\right)^{\ell}\\
    &\lesssim \e^{-c'\frac{\ell}{d\beta}} +\labs{A} \ell \e^{-\ell} \lesssim \e^{-c'\frac{\ell}{d\beta}} +\labs{A} 2^{-\ell}\,, \end{align} 
{for some absolute constants $c,c'>0$}. The last line sets $T = \frac{\ell}{4\e^2d\beta} -1\ge 0$ and updates the constants to conclude the proof.
\end{proof}

The proof of \autoref{lem:quasilocal} then follows from combining the bounds derived above:

\begin{proof}[Proof of \autoref{lem:quasilocal}]
Using the bounds derived in \autoref{lem.dissipaticepartboundl} and \autoref{cor:localjumps}, we get that 
\begin{align*}
\|\mathcal{L}_{A,\ell}^\dagger-\mathcal{L}_A^\dagger\|_{\infty-\infty}\le \|\mathcal{D}^\dagger-\mathcal{D}^\dagger_{\ell}\|_{\infty-\infty}+2\|\vB-\vB_\ell\|&\lesssim |A|\,\Big(\e^{-\sigma^2\ell^2/2d^2}+\e^{-\ell}+\e^{-c'\frac{\ell}{d\beta}}+2^{-\ell}\Big)\\
&\lesssim |A|\Big(\e^{-c'\frac{\ell}{d\beta}} + 2^{-\ell}\Big)\,.
\end{align*}

\end{proof}

\section{Improving the bounds with a local gap condition}\label{sec:improvegap}

In this appendix, we propose a method to improve over the bounds derived in our main results \autoref{cor:CMI} and \autoref{cor:samplingresult} under an additional \textit{local gap} condition.

\begin{defn}[Local gap] 
We say a $\vrho$-detailed balanced Lindbladian $\CL$ is $\lambda$-locally-gapped if
    \begin{align*}
-\lambda\,\langle \vX,\mathcal{L}^\dagger(\vX)\rangle_{\vrho}\le \langle \vX,\mathcal{L}^{\dagger2}(\vX)\rangle_{\vrho} \quad \text{for each operator}\quad \vX.
    \end{align*} If the above holds for any restricted Gibbs state $\vrho^X$, with $X\subset \Lambda$, we say the pair of Hamiltonian and set of jumps $(\vH,\{\vA^a\})$ is $\lambda$-uniformly-locally-gapped at inverse temperature $\beta$. 
\end{defn}
The above unconditionally holds for $(\vH,\{\vA^a\})$ being a commuting Hamiltonian and local jumps $P_A^1$ on a region A with a local gap independent of the global system size $\lambda \ge f(\labs{A})>0$. For general noncommutative Hamiltonians (with local jumps), we do not know of any a priori bound on the local gap, even assuming high temperature. Nevertheless, we may, for now, play around with the consequences of such strong assumptions, and leave for future work to explore suitable local gap conditions. 
\begin{lem}[exponential decay of Dirichlet form] Suppose a Lindbladian is $\lambda$-gapped and consider the limit
\begin{align}
    \CP:=\lim_{s\to\infty}\e^{s\mathcal{L}}.
\end{align}
Then for any $t>0$,
\begin{align*}
\| \e^{\CL^{\dagger} t}(\vX)-\CP^\dagger(\vX)\|_{\vrho} \le \e^{-t\lambda }\|\vX\|_{\vrho}.
\end{align*}
\end{lem}

\begin{lem} Suppose the Lindbladian is $\lambda$-locally-gapped. Then, the Dirichlet form decays exponentially: for any $t>1$,
\begin{align}
    \CE(\e^{\CL^{\dagger}t}[\vX],\e^{\CL^{\dagger}t}[\vX]) \le \e^{-2\lambda (t-1)}\norm{\vX}_{\vrho}^2.
\end{align}

\end{lem}
\begin{proof}
    Consider the equivalent linear algebra statement for PSD operator $\vV = \sum_i v_i\ket{v_i}\bra{v_i}$, and for any state $|\psi\rangle$,
    \begin{align}
    \langle \psi|{\vV \e^{-2\vV t}}|\psi\rangle \le \max_{i} v_i \e^{-2v_i t} \le {\max_{i, v_i\ne 0}}\,\e^{2v_i}\e^{-2v_it}   
    \end{align}
 using that $v \le \e^{2v}$ for any $v>0.$  
\end{proof}

Similarly, the cost for the patching scheme \autoref{cor:samplingresult} for preparing Gibbs states over a $D$-dimensional lattice is dominated by the maps $\mathcal{R}_{A,t^*,\ell}$ for some log-size regions $A$ for a quasi-polynomial time $t^*=\e^{\Theta(\log^D(n/\epsilon))}$. Adapting from the patching argument of~\autoref{cor:samplingresult} with the faster decay of Dirichlet form arising from the local gap condition, we arrive at a much improved cost. Roughly, one could boost any polynomially small local gap to a preparation scheme with poly-logarithmic overheads.
\begin{cor}[polynomial local gap and efficient preparation]
    Assume that the Lindbladian with single jumps
\begin{align}
    \CL = \sum_{i\in P_A^1} \CL_i 
\end{align}
is $c\labs{A}^{-c'}$-uniformly-locally-gapped for some constants $c, c'>0.$ Then, the Gibbs state can be prepared with gate complexity
\begin{align}
\CO( n\poly\log(n/\epsilon)).     
\end{align}
\end{cor}
By the same reasoning, the local Markov property of \autoref{cor:CMI} can be upgraded into a global one assuming the local gap condition:
\begin{cor}
Again, Assume that the Lindbladian with single jumps is $c |A|^{-c'}$-uniformly-locally-gapped for some constants $c,c'>0$. Then, for any tripartition of the system $\Lambda=A\sqcup B\sqcup C$ the quantum conditional mutual information (QCMI) evaluated at $\vrho_\beta$ satisfies
\begin{align}\label{CMIdecayimproved}
    I(A:C|B)_{\vrho_\beta}\le  \poly(|A|,|C|)\,\exp\bigg(-\frac{\operatorname{dist}(A,C)}{\xi}\bigg)\,,
\end{align}
for some polynomial and constant $\xi$ depending on $\beta, d,c,c'$.
\end{cor}

\bibliographystyle{alphaUrlePrint.bst}
\bibliography{ref}

\end{document}